\documentclass{scrartcl}

\usepackage[T1]{fontenc}
\usepackage[utf8]{inputenc}
\usepackage{amsmath, amsthm, amssymb}
\usepackage{dsfont}
\usepackage{graphicx}
\usepackage[format=plain]{caption}
\usepackage{capt-of}
\usepackage{subcaption}
\usepackage{bm}

\newcommand{\norm}[1]{\left\| #1 \right\|}  
\newcommand{\scprd}[1]{\left\langle #1 \right\rangle}  
\newcommand{\ddt}[1]{\frac{d} {dt} #1}   
\newcommand{\ddtl}[1]{\frac{d^l}{dt^l} #1}  
\newcommand{\N}{\mathbb{N}}  

\newcommand{\R}{\mathbb{R}}
\newcommand{\C}{\mathbb{C}}

\newcommand{\ran}{\operatorname{ran}}
\newcommand{\rk}{\operatorname{rk}}
\newcommand{\supp}{\operatorname{supp}}
 
\newcommand{\dist}{\operatorname{dist}}

\newcommand{\eps}{\varepsilon}
\renewcommand{\phi}{\varphi}

\newcommand{\ol}{\overline}

\renewcommand{\Im}{\operatorname{Im}}

\numberwithin{equation}{section}

\newtheorem{thm}{Theorem}[section]
\newtheorem{cor}[thm]{Corollary}
\newtheorem{prop}[thm]{Proposition}
\newtheorem{lm}[thm]{Lemma}
\newtheorem{cond}[thm]{Condition}

\theoremstyle{definition} 
\theoremstyle{definition}

\title{Adiabatic theorems for general linear operators with time-dependent domains}

\author{Jochen Schmid\\  
\small Institut f\"ur Mathematik, Universit\"at W\"urzburg, 97074 W\"urzburg, Germany\\
\small jochen.schmid@mathematik.uni-wuerzburg.de}   

\date{}

\begin{document}

\maketitle

\begin{abstract}
\small{ \noindent 
We establish adiabatic theorems with and without spectral gap condition for general -- typically dissipative -- 
linear operators $A(t): D(A(t)) \subset X \to X$ with time-dependent domains $D(A(t))$ in some Banach space $X$. In these theorems, we do not require the considered spectral values $\lambda(t)$ of $A(t)$ to be (weakly) semisimple. We then apply our general theorems to the special case of skew-adjoint operators $A(t) = 1/i A_{a(t)}$ defined by symmetric sesquilinear forms $a(t)$ 
and thus generalize, in a very simple way, the only adiabatic theorem for operators with time-dependent domains known so far. 
}
\end{abstract}

{ \small \noindent \emph{Subject classification (2010) and key words:} 34E15, 34G10, 35Q41, 47D06
\\
Adiabatic theorems for general linear operators, dissipative operators, time-dependent domains, non-semisimple spectral values, spectral gap, time-dependent symmetric sesquilinear forms
}

\section{Introduction} \label{sect:intro}

%



Adiabatic theory -- or, more precisely, time-adiabatic theory for linear operators with time-dependent domains -- 
is concerned with slowly time-varying systems described by 
evolution equations
\begin{align} \label{eq: awp, adtheorie 0}
x' = A(\eps s) x \quad (s \in [s_0,1/\eps]) \quad \text{and} \quad x(s_0) = y, 
\end{align}
where $A(t): D(A(t)) \subset X \to X$ for $t \in [0,1]$ is a densely defined closed linear operator with domain $D(A(t))$ in a Banach space $X$ and where $\eps \in (0,\infty)$ is some (small) slowness parameter.
Smaller and smaller values of $\eps$ mean that $A(\eps s)$ depends more and more slowly on time $s$ or, in other words, that the typical time 
where $A(\eps \,.\,)$ varies 
appreciably gets larger and larger. 
%
Such slowly time-varying systems arise, for instance, when an electric or magnetic potential is slowly switched on or in approximate molecular dynamics (in the context of the Born--Oppenheimer approximation).
%
It is common and convenient in adiabatic theory to rescale time as $t = \eps s$ and to consider the equivalent rescaled evolution equations 
\begin{align} \label{eq: awp, adtheorie}
x' = \frac 1 \eps A(t) x \quad (t \in [t_0,1]) \quad \text{and} \quad x(t_0) = y
\end{align}
with initial times $t_0 \in (0,1]$ and initial values $y \in D(A(t_0))$. 
It is further assumed 
that these evolution equations are well-posed, that is, for every initial time $t_0 \in (0,1]$ and every initial value $y \in D(A(t_0))$ the initial value problem~\eqref{eq: awp, adtheorie} has a unique classical solution $x_{\eps}(\,.\,,t_0,y)$ and $x_{\eps}(\,.\,,t_0,y)$ 
continuously depends on $t_0$ and $y$. A bit more concisely and conveniently, the well-posedness of~\eqref{eq: awp, adtheorie} can be characterized by the existence of a unique so-called evolution system $U_{\eps}$ for $\frac 1 \eps A$ on the spaces $D(A(t))$, that is, a two-parameter family of bounded solution operators $U_{\eps}(t,t_0)$ in $X$ determined by $U_{\eps}(t,t_0)y = x_{\eps}(t,t_0,y)$ for $y \in D(A(t_0))$ and $t_0 \le t$. 
\smallskip


Adiabatic theory is further concerned with curves of spectral values $\lambda(t) \in \sigma(A(t))$, mostly eigenvalues, 
of the operators $A(t)$. In the classical special case of skew-adjoint operators $A(t)$ (that is, operators of the form $1/i$ times a self-adjoint operator $A_0(t)$), these spectral values $\lambda(t) = 1/i \, \lambda_0(t)$ could correspond to the ground-state energy $\lambda_0(t)$ of $A_0(t)$, for instance.
If $\lambda(t)$ is isolated in the spectrum $\sigma(A(t))$ of $A(t)$ for every $t \in [0,1]$, one speaks of a spectral gap. And such a spectral gap, in turn, is called uniform or non-uniform depending on whether or not 
\begin{align} \label{eq: ausglage und -frage, glm sl}
\inf_{t \in [0,1]} \dist \big( \lambda(t), \sigma(A(t))\setminus \{\lambda(t)\}  \big) > 0.
\end{align}
Some typical spectral situations are illustrated below for the special case of skew-adjoint operators $A(t)$: the spectrum $\sigma(A(t))$ is plotted on the vertical axis $i \R$ against the horizontal $t$-axis and the red line represents the considered spectral values $\lambda(t)$. In the first two figures, we have a spectral gap which is uniform in the first and non-uniform in the second picture. And the third figure depicts a situation without spectral gap.

\vfill
\begin{figure}[htbp]%
\centering
\begin{subfigure}[b]{0.3\textwidth}
\centering
\includegraphics[width=\columnwidth]{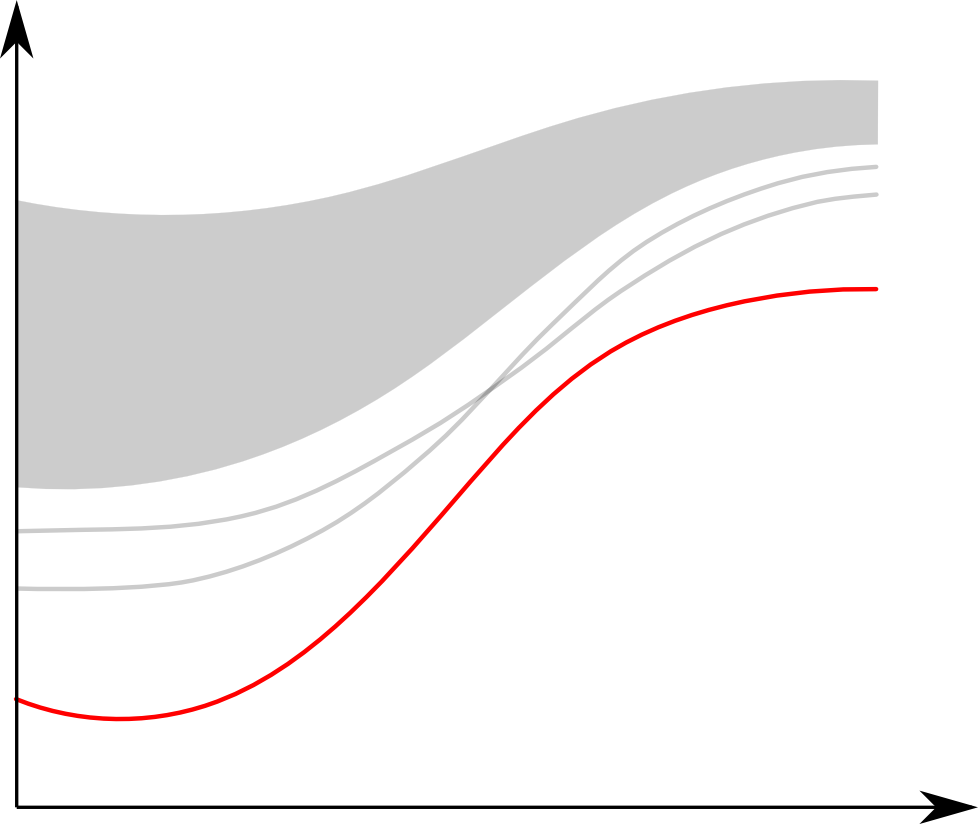}
\end{subfigure}
\hfill
\begin{subfigure}[b]{0.3\textwidth}
\centering
\includegraphics[width=\columnwidth]{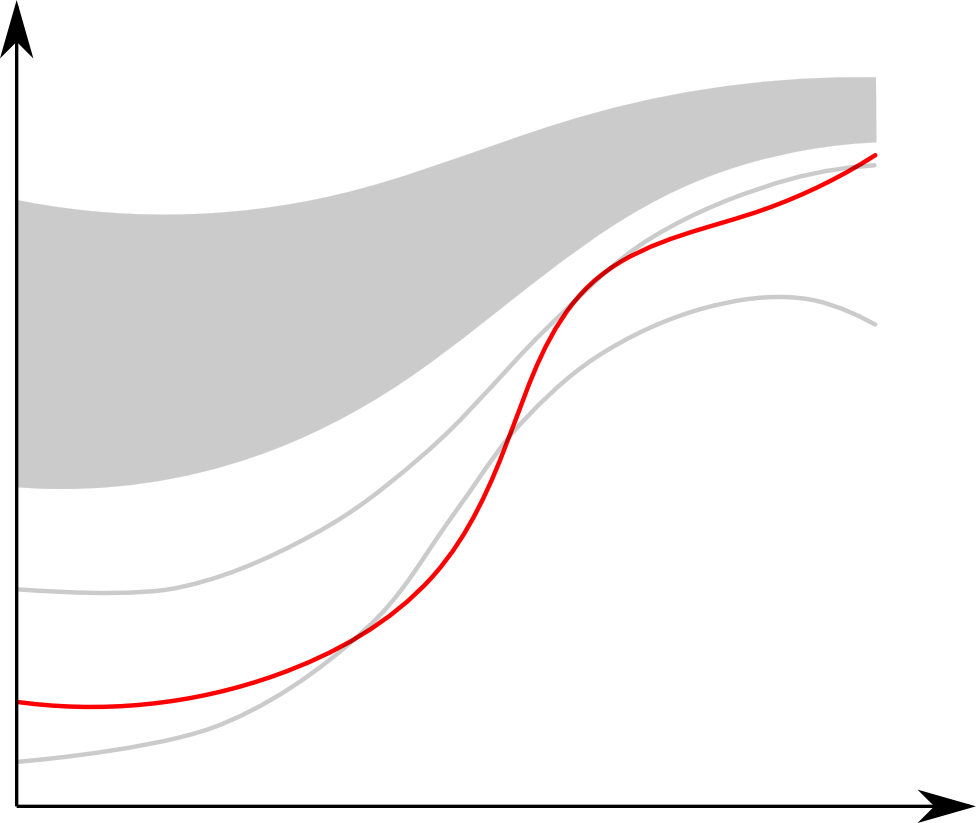}
\end{subfigure}
\hfill
\begin{subfigure}[b]{0.3\textwidth}
\centering
\includegraphics[width=\columnwidth]{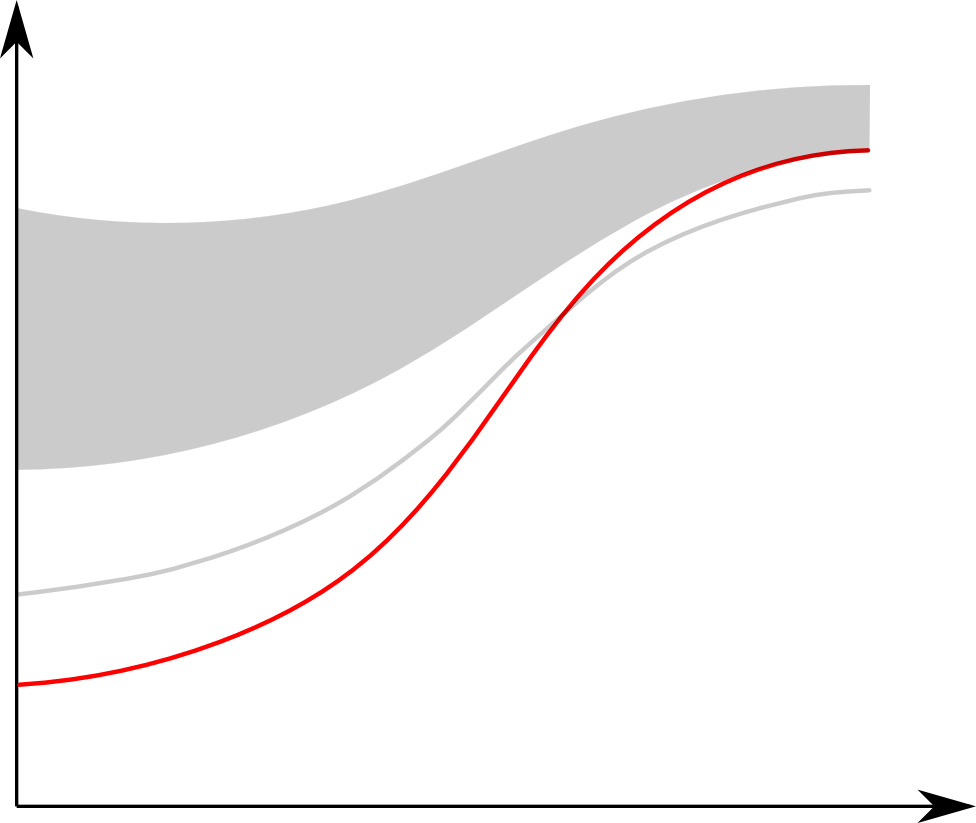}
\end{subfigure}
\end{figure}

%


What adiabatic theory is interested in is how certain distinguished solutions to~\eqref{eq: awp, adtheorie} behave in the singular limit where the slowness parameter $\eps$ tends to $0$. In more specific terms, the basic goal of adiabatic theory can be described -- for skew-adjoint and then for general operators -- as follows. 
In the special case of skew-adjoint operators $A(t)$, one wants to show that for small values of $\eps$ and every $t$ the solution operator $U_{\eps}(t,0)$ 
takes eigenvectors of $A(0)$ corresponding to $\lambda(0)$ into eigenvectors of $A(t)$ corresponding to $\lambda(t)$ -- up to small errors in $\eps$. Shorter and more precisely, one wants to show that 
\begin{align}  \label{eq: aussage des adsatzes, schiefsa}
(1-P(t)) U_{\eps}(t,0) P(0) \longrightarrow 0 \qquad (\eps \searrow 0)
\end{align}
for all $t \in [0,1]$, where $P(t)$ for (almost) every $t$ is the canonical spectral projection of $A(t)$ corresponding to $\lambda(t)$. 
It is defined via the spectral measure $P^{A(t)}$ of $A(t)$, namely $P(t) = P^{A(t)}(\{\lambda(t)\})$, and it is the orthogonal projection yielding the decomposition of $X$ into $P(t)X = \ker(A(t)-\lambda(t))$ and $(1-P(t))X = \ol{\ran}(A(t)-\lambda(t))$.
In the case of general operators $A(t)$, one again wants to show that
\begin{align}  \label{eq: aussage des adsatzes}
(1-P(t)) U_{\eps}(t,0) P(0) \longrightarrow 0 \qquad (\eps \searrow 0)
\end{align}
for all $t \in [0,1]$, where now $P(t)$ for (almost) every $t$ is a suitable general 
spectral projection of $A(t)$ corresponding to $\lambda(t)$. 
In the case with spectral gap, suitable 
spectral projections are 
the so-called associated projections, which yield the decomposition
\begin{align} \label{eq: zerl, sl}
P(t)X = \ker(A(t)-\lambda(t))^{m(t)} \quad \text{and} \quad (1-P(t))X = \ran(A(t)-\lambda(t))^{m(t)}
\end{align}
for some $m(t) \in \N$ provided $\lambda(t)$ is a pole of $(\,.\,-A(t))^{-1}$.
In the case without spectral gap, suitable 
spectral projections are 
the so-called weakly associated projections, which yield the decomposition
\begin{align}  \label{eq: zerl, ohne sl}
P(t)X = \ker(A(t)-\lambda(t))^{m(t)} \quad \text{and} \quad (1-P(t))X = \ol{\ran}(A(t)-\lambda(t))^{m(t)}
\end{align}
for some $m(t) \in \N$.
An adiabatic theorem is now simply a theorem that gives conditions on $A(t)$, $\lambda(t)$, $P(t)$ under which the convergence~\eqref{eq: aussage des adsatzes} holds true. 
%
%
Sometimes, we will distinguish quantitative and qualitative adiabatic theorems depending on whether they 
give information on the rate of convergence in~\eqref{eq: aussage des adsatzes} or not. Specifically, if the rate of convergence in~\eqref{eq: aussage des adsatzes} can be shown to be of polynomial order $\eps^n$ or even exponential order $e^{-c/\eps}$ in $\eps$, we will speak of an adiabatic theorem of higher order. 
\smallskip


Adiabatic theory has a long history going back to the first days of quantum theory and many authors have contributed to it since then. 
In the first decades after 1928, all adiabatic theorems were exclusively concerned with skew-adjoint operators $A(t)$ and until 1998 they all required a spectral gap condition. See, for instance, \cite{BornFock28}, \cite{Kato50}, \cite{JoyePfister93}, \cite{Nenciu93} for the case with spectral gap and \cite{AvronElgart99}, \cite{Bornemann98}, \cite{Teufel01} \cite{FishmanSoffer16}, for instance, for the case without spectral gap.
In the last decade, various adiabatic theorems for more general operators $A(t)$ have been established and again, just like in the special case of skew-adjoint operators, the case with spectral gap has been treated first. A major motivation for these general adiabatic theorems 
comes from applications to open quantum systems which, unlike closed quantum systems, cannot be described by skew-adjoint operators anymore. 
See, for instance, \cite{AbouSalem07}, \cite{Joye07}, \cite{HansonJoyePautratRaquepas17} for the case with spectral gap and \cite{AvronGraf12}, \cite{dipl}, \cite{JaksicPillet14}, for instance, for the case without spectral gap.
A detailed historical overview can be found in~\cite{diss}, for instance.
\smallskip
 
%
So far, however, almost all adiabatic theorems from the literature require the domains of the operators $A(t)$ to be time-independent, that is, $D(A(t)) = D$ for all $t \in [0,1]$. In fact, there is only one adiabatic theorem, namely the one from~\cite{Bornemann98}, where the operators $A(t)$ are allowed to have time-dependent domains. In this result, skew-adjoint operators $A(t) = 1/i A_{a(t)}$ defined by symmetric sesquilinear forms $a(t)$ with time-independent (form) domain are considered along with spectral values $\lambda(t)$ that are assumed to belong to the discrete spectrum of $A(t)$ (hence, in particular, isolated).
Such operators arise, for instance, as Schr\"odinger operators $-\Delta + V(t)$ (sum in the form sense) with time-dependent Rollnik potentials $V(t)$. 
\smallskip


In this paper, we establish adiabatic theorems with and without spectral gap condition (including an adiabatic theorem of higher order) for general linear operators $A(t): D(A(t)) \subset X \to X$ with time-dependent domains $D(A(t))$. 
In particular, we do not require the considered spectral values $\lambda(t)$ to be semisimple (case with spectral gap) or weakly semisimple (case without spectral gap), that is, we do not require $m(t) = 1$ in the decomposition~\eqref{eq: zerl, sl} or~\eqref{eq: zerl, ohne sl}, respectively.
With these theorems, we extend the adiabatic theorems from~\cite{zeitunabh} and~\cite{JoyePfister93} and in our proofs we can closely follow the proofs from~\cite{zeitunabh} and~\cite{JoyePfister93}. 
We then apply our general theorems 
to the special case of skew-adjoint operators $A(t) = 1/i A_{a(t)}$ defined by closed semibounded symmetric sesquilinear forms $a(t)$ with time-independent (form) domain. In that way, we obtain among other things the following adiabatic theorem without spectral gap condition, which generalizes the adiabatic theorem from~\cite{Bornemann98}. In simplified form, it can be formulated as follows (with $I := [0,1]$). See~\cite{Schmid15qmath}. If $A(t)  = 1/i A_{a(t)}$ and $a(t)$ are as above and $\lambda(t)$ for every $t \in I$ is an eigenvalue of $A(t)$ and if $P(t)$ for almost every $t \in I$ is weakly associated with $A(t)$ and $\lambda(t)$ and of finite rank, then -- under suitable regularity assumptions -- 
one has 
\begin{align}
\sup_{t \in I} \norm{(1-P(t))U_{\eps}(t,0)P(0)}, \quad \sup_{t \in I} \norm{P(t)U_{\eps}(t,0)(1-P(0))} \longrightarrow 0
\end{align}
as $\eps \searrow 0$. In the above relation, $U_{\eps}$ denotes the evolution system for $\frac{1}{\eps} A$ on the spaces $D(A(t))$. 
Apart from yielding more general results, our strategy of proof is also considerably simpler than the -- completetly different -- method of proof from~\cite{Bornemann98}. 
\smallskip


In Section~\ref{sect:preliminaries} we recall 
the most important preliminaries needed for our adiabatic theorems, namely on well-posedness and evolution systems (Section~\ref{sect:wp-evol-syst}), on associated and weakly associated projections (Section~\ref{sect:spectral-proj}), and on adiabatic evolution systems (Section~\ref{sect:ad-evol}). Section~\ref{sect:ad-thms-general} contains our adiabatic theorems for general operators $A(t)$, while Section~\ref{sect:ad-thms-skew-adjoint} is devoted to the special case of skew-adjoint operators $A(t) = 1/i A_{a(t)}$ defined by symmetric sesquilinear forms $a(t)$. 
%
%
%
%
In the entire paper, we will use the following notational conventions. $X$, $Y$, $Z$ will denote Banach spaces, $H$ 
a Hilbert space over $\C$ and $\norm{\,.\,}_{X,Y}$ will stand for the operator norm on $L(X,Y)$, the space of bounded linear operators from $X$ to $Y$. If $X=Y$, we will usually simply write $\norm{\,.\,}$ for $\norm{\,.\,}_{X,X}$.
Also, we abbreviate 
\begin{align*}
I := [0,1] \qquad \text{and} \qquad \Delta := \{(s,t) \in I^2: s \le t\}
\end{align*}
and for evolution systems $U$ defined on $\Delta$ we will write $U(t) := U(t,0)$ for brevity. 
And finally, whenever a family of linear operators $A(t): D(A(t)) \subset X \to X$ is given and the evolution  system for $\frac{1}{\eps} A$ on the spaces $D(A(t))$ exists, it will be denoted by $U_{\eps}$. 

\section{Some preliminaries} \label{sect:preliminaries}

\subsection{Well-posedness and evolution systems} \label{sect:wp-evol-syst}

In this section, we recall from~\cite{EngelNagel} the concepts of well-posedness and (solving) evolution systems for non-autonomous linear evolution equations 
\begin{align} \label{eq: ivp, vorber}
x' = A(t)x  \quad (t \in [s,b]) \quad \text{and} \quad  x(s) = y
\end{align}
with densely defined linear operators $A(t): D(A(t)) \subset X \to X$  ($t \in [a,b]$) and initial values $y \in D(A(s))$ at initial times $s \in [a,b)$. 
\smallskip

Well-posedness of evolution equations~\eqref{eq: ivp, vorber} means, of course, something like unique (classical) solvability with continuous dependence on the initial data. 
In precise terms, the initial value problems~\eqref{eq: ivp, vorber} for $A$ are called \emph{well-posed on (the spaces) $D(A(t))$} if and only if there exists a \emph{(solving) evolution system for $A$ on (the spaces) $D(A(t))$}. 
Such an evolution system for $A$ on $D(A(t))$ is, by definition, 
a family $U$ of bounded operators $U(t,s)$ in $X$ for $(s,t) \in \Delta_J := \{ (s,t) \in J^2: s \le t \}$ such that 
\begin{itemize}
\item [(i)] for every $s \in [a,b)$ and $y \in D(A(s))$, the map $[s,b] \ni t \mapsto U(t,s)y$ is a continuously differentiable solution to the initial value problem~\eqref{eq: ivp, vorber}, that is,  
a continuously differentiable map $x: [s,b] \to X$ such that $x(t) \in D(A(t))$ and $x'(t) = A(t) x(t)$ for all $t \in [s,b]$ and $x(s) = y$,
\item[(ii)] $U(t,s) U(s,r) = U(t,r)$ for all $(r,s), (s,t) \in \Delta_J$ 
and $\Delta_J \ni (s,t) \mapsto U(t,s)x$ is continuous for all $x \in X$.
\end{itemize}
%
\smallskip

If, for a given family $A$ of densely defined operators $A(t): D(A(t)) \subset X \to X$, there exists any solving evolution system, then it is already unique. In order to see this we need the following simple lemma, which will always be used when the difference of two evolution systems has to be dealt with. 

\begin{lm} \label{lm: zeitentw rechtsseit db}
Suppose $A(t): D(A(t)) \subset X \to X$ is a densely defined linear operator for every $t \in J$ and that $U$ is an evolution system for $A$ on $D(A(t))$. 
Then, for every $s_0 \in [a,t)$ and every $x_0 \in D(A(s_0))$, the map $[a,t] \ni s \mapsto U(t,s)x_0$ is right differentiable at $s_0$ with right derivative $-U(t,s_0) A(s_0) x_0$. 
\end{lm}


\begin{proof}
Since $U(t,s)U(s,r) = U(t,r)$ for $(r,s), (s,t) \in \Delta_J$ and since $\Delta_J \ni (s,t) \mapsto U(t,s)$ is strongly continuous, we obtain for every $s_0 \in [a,t)$ and $x_0 \in D(A(s_0))$ that
\begin{align*}
\frac{ U(t,s_0+h)x_0 - U(t,s_0)x_0 }{h} &= -U(t,s_0+h) \frac{ U(s_0+h,s_0)x_0 - x_0 }{h} \\
&\longrightarrow -U(t,s_0)A(s_0)x_0 
\end{align*}
as $h \searrow 0$, as desired. 
\end{proof}

\begin{cor} \label{cor: zeitentw eind} 
Suppose $A(t): D(A(t)) \subset X \to X$ is a densely defined linear operator for every $t \in J$. If $U$ and $V$ are two evolution systems for $A$ on $D(A(t))$, then $U = V$.
\end{cor}

\begin{proof}
If $U$ and $V$ are two evolution systems for $A$ on the spaces $D(A(t))$, then for every $(s,t) \in \Delta_J$ with $s < t$ and $y \in D(A(s))$ the map $[s,t] \ni \tau \mapsto U(t,\tau)V(\tau,s)y$ is continuous and right differentiable with vanishing right derivative by virtue of Lemma~\ref{lm: zeitentw rechtsseit db}. With the help of Corollary~2.1.2 of~\cite{Pazy} it then follows that
\begin{align*}
V(t,s)y - U(t,s)y = U(t,\tau)V(\tau,s)y \big|_{\tau=s}^{\tau=t} = 0,
\end{align*}
which by the density of $D(A(s))$ in $X$ implies $U(\,.\,,s) = V(\,.\,,s)$. Since $s$ was arbitrary in $[a,b)$ we obtain $U = V$, as desired.
\end{proof}

We will also need the following perturbation result. It provides an estimate for the evolution system $V$ for a perturbed family $A + B$ based on a corresponding estimate for the evolution $U$ for the unperturbed family $A$, provided both these evolutions  exist. 

\begin{prop} \label{prop: störreihe für gestörte zeitentw}
Suppose that $A(t): D(A(t)) \subset X \to X$ is a densely defined linear operator for every $t \in I$ and that $t \mapsto B(t) \in L(X)$ is strongly continuous. Suppose further that there is an evolution system $U$ for $A$ on $D(A(t))$ and an evolution system $V$ for $A+B$ on $D(A(t))$ and that $\norm{ U(t,s) } \le M e^{ \omega (t-s) }$ for all $(s,t) \in \Delta$ and some $M \in [1,\infty)$ and $\omega \in \R$. Then
\begin{align*}
\norm{ V(t,s) } \le M e^{ (\omega + M b)(t-s) }
\end{align*}
for all $(s,t) \in \Delta$, where $b := \sup_{t \in I} \norm{ B(t)}$. 
\end{prop}

\begin{proof}
Since for all $x \in D(A(s))$ ($s \in [0,1)$ fixed) $[s,t] \ni \tau \mapsto U(t,\tau) V(\tau,s)x$ is continuous and right differentiable (Lemma~\ref{lm: zeitentw rechtsseit db}) and since the right derivative $\tau \mapsto U(t,\tau) B(\tau) V(\tau,s)x$ is continuous, it follows from Corollary~2.1.2 of~\cite{Pazy} that 
\begin{align}  \label{eq:int-gl-gestoerte-zeitentw}
U(t,s)x - V(t,s)x = U(t,\tau) V(\tau,s)x \big|_{\tau = s}^{\tau = t} = \int_s^t U(t,\tau) B(\tau) V(\tau,s)x \, d\tau
\end{align}
for all $t \in [s,1]$. A simple Gronwall argument now yields the asserted estimate. 
\end{proof}

In the situation of the above proposition, one also obtains a perturbation series expansion for $V$ by inserting the integral representation of $V$ from~\eqref{eq:int-gl-gestoerte-zeitentw} into the right-hand side of~\eqref{eq:int-gl-gestoerte-zeitentw} again and again. Specifically, 
$V(t,s) = \sum_{n=0}^{\infty} V_n(t,s)$, where $V_0(t,s) := U(t,s)$ and 
\begin{align*}
V_{n+1}(t,s)x := \int_s^t U(t,\tau) B(\tau) V_n(\tau, s)x \, d\tau \qquad (x \in X)
\end{align*}
for $n \in \N \cup \{ 0 \}$.
In view of this explicit representation of $V$ in terms of $U$ and $B$, one might think that the existence of the evolution system $U$ for $A$ on $D(A(t))$ and the strong continuity of $t \mapsto B(t)$ alone already guarantee that the evolution system for $A+B$ exists on $D(A(t))$ (and is given by the above perturbation series) 
-- but this is not true. See the examples in~\cite{Phillips53} or~\cite{SchmidGriesemer16}, for instance.

\subsection{Spectral projections for general linear operators} \label{sect:spectral-proj}

In this section we recall from~\cite{zeitunabh} suitable 
notions of spectral projections for general linear operators, namely the notion of associated projections (which is completely canonical) 
and the notion of weakly associated projections (which --~for non-normal, or at least, non-spectral operators -- is not canonical). 
\smallskip

Suppose $A: D(A) \subset X \to X$ is a densely defined closed linear operator with $\rho(A) \ne \emptyset$, $\sigma \ne \emptyset$ is a compact isolated subset of $\sigma(A)$, $\lambda$ a not necessarily isolated spectral value of $A$, and $P$ a bounded projection in $X$.
We then say, following~\cite{TaylorLay80}, that 
\emph{$P$~is associated with $A$ and $\sigma$} if and only if 
$P$ commutes with $A$, 
$P D(A) = P X$ and
\begin{align*}
\sigma(A|_{PD(A)}) = \sigma  \text{ \, whereas \, }  \sigma(A|_{(1-P)D(A)}) = \sigma(A) \setminus \sigma.
\end{align*}
We say that \emph{$P$~is weakly associated with $A$ and $\lambda$} if and only if 
$P$ commutes with $A$, $P D(A) = P X$ and 
\begin{gather*}
A|_{PD(A)} - \lambda \text{\: is nilpotent whereas \:} A|_{(1-P)D(A)} - \lambda \text{\: is injective and} \\
\text{has dense range in } (1-P)X.
\end{gather*}
If above the order of nilpotence is at most $m$, we will often, 
more precisely, speak of $P$ as being \emph{weakly associated with $A$ and $\lambda$ of order $m$}. 
%
Also, we call $\lambda$ a \emph{weakly semisimple eigenvalue of $A$} if and only if $\lambda$ is an eigenvalue and there is a projection $P$ weakly associated with $A$ and $\lambda$ of order $1$. In this context, 
recall that $\lambda$ is called a \emph{semisimple eigenvalue of $A$} if and only if it is a pole of the resolvent map $(\,.\,-A)^{-1}$ of order $1$ (which is then automatically an eigenvalue by~\eqref{eq: zerl von X, lambda isoliert} below). Also, a semisimple eigenvalue is called \emph{simple} if and only if its geometric 
multiplicity is $1$.
\smallskip

In our adiabatic theorems below, we will continually use the following central facts about 
associatedness and weak associatedness, concerning the question of existence and uniqueness of (weakly) associated projections (for given operators $A$ and spectral values $\lambda$) and the question of describing (in terms of $A$ and $\lambda$) the 
subspaces into which a (weakly) associated projection decomposes the base space $X$. See~\cite{zeitunabh} (Section~2.4) for proofs of these facts.

\begin{thm} \label{thm: central properties associatedness}
Suppose $A: D(A) \subset X \to X$ is a densely defined closed linear operator with $\rho(A) \ne \emptyset$ 
and $\emptyset \ne \sigma \subset \sigma(A)$ is compact. 
If $\sigma$ is 
isolated in $\sigma(A)$, then there exists a unique projection $P$ associated with $A$ and $\sigma$, 
namely 
\begin{align*}  
P := \frac{1}{2 \pi i} \int_{\gamma} (z-A)^{-1} \, dz,
\end{align*}
where $\gamma$ is a cycle in $\rho(A)$ with indices $\operatorname{n}(\gamma, \sigma) = 1$ and $\operatorname{n}(\gamma, \sigma(A) \setminus \sigma) = 0$.
If $P$ is associated with $A$ and $\sigma = \{ \lambda \}$ and $\lambda$ is a pole of $(\,.\,-A)^{-1}$ 
of order $m$, then
\begin{align}  \label{eq: zerl von X, lambda isoliert}
PX = \ker(A-\lambda)^k \quad \text{and} \quad (1-P)X = \ran(A-\lambda)^k 
\end{align}
for all $k \in \N$ with $k \ge m$.
\end{thm}

\begin{thm} \label{thm: typ mögl für PX und (1-P)X}
Suppose $A: D(A) \subset X \to X$ is a densely defined closed linear operator with $\rho(A) \ne \emptyset$ 
and $\lambda \in \sigma(A)$. 
If $\lambda$ is non-isolated in $\sigma(A)$, then in general there exists no projection $P$ weakly associated with $A$ and $\lambda$, but if such a projection exists it is already unique. 
If $P$ is weakly associated with $A$ and $\lambda$ of order $m$, then
\begin{align}  \label{eq: zerl von X, lambda nicht isoliert}
PX = \ker(A-\lambda)^k \quad \text{and} \quad (1-P)X = \overline{ \ran}(A-\lambda)^k 
\end{align}
for all $k \in \N$ with $k \ge m$.
\end{thm}


Since for given operators $A$ and spectral values $\lambda$ there will in general exist no projection weakly associated with $A$ and $\lambda$, 
it is important to have criteria for the existence of weakly associated projections. See~\cite{zeitunabh} (Section~2.4) for two such criteria --  one for spectral operators $A$ and one for generators $A$ of bounded semigroups and spectral values $\lambda \in i \R$. With regard to our adiabatic theorems for skew-adjoint operators from Section~\ref{sect:ad-thms-skew-adjoint}, 
the following special case of the two criteria from~\cite{zeitunabh} is particularly important. 
It shows that for skew-adjoint operators, the existence issues for weakly associated projections disappear.   

\begin{prop}
If $A: D(A) \subset X \to X$ is a skew-adjoint operator with spectral measure $P^{A}$ and $\lambda \in \sigma(A)$, then there exists a (unique) projection $P$ weakly associated with $A$ and $\lambda$ and it is given by $P = P^{A}(\{\lambda\})$. 
If, in addition, $\lambda$ is isolated in $\sigma(A)$, then the projection associated with $A$ and $\lambda$ is given by $P^{A}(\{\lambda\})$ as well.  
\end{prop}

\begin{proof}
With the standard theory of self-adjoint (or normal) operators, it immediately follows that $P^{A}(\{\lambda\})$ is a projection that is weakly associated with $A$ and $\lambda$. Since, by the previous theorem, weakly associated projections are unique as soon as they exist, the uniqueness statement is clear as well. 
If, in addition, $\lambda$ is isolated in $\sigma(A)$, then it is  well-known that 
\begin{align*}
\frac{1}{2 \pi i} \int_{\gamma} (z-A)^{-1} \, dz = P^{A}(\{\lambda\}),
\end{align*}
which proves the last part of the proposition. 
\end{proof}


In the proof of our adiabatic theorem without spectral gap condition, we will also need that in reflexive spaces weak associatedness carries over to the dual operators -- provided that some core condition is satisfied, which is the case for semigroup generators, for instance (Proposition~II.1.8 of~\cite{EngelNagel}). See~\cite{zeitunabh} (Section~2.4) for a proof.

\begin{prop} \label{prop: schwache assoziiertheit, dual}
Suppose $A: D(A) \subset X \to X$ is a densely defined closed linear operator in the reflexive space $X$ such that $\rho(A) \ne \emptyset$ and $D(A^k)$ is a core for $A$ for all $k \in \N$. If $P$ is weakly associated with $A$ and $\lambda \in \sigma(A)$ of order $m$, then $P^*$ is weakly associated with $A^*$ and $\lambda$ of order $m$. 
\end{prop}

\subsection{Adiabatic evolutions}  \label{sect:ad-evol}

We say that an evolution system for a family $A$ of linear operators $A(t): D(A(t)) \subset X \to X$ is \emph{adiabatic w.r.t.~a family $P$ of bounded projections $P(t)$ in $X$} if and only if $U(t,s)$ for every $(s,t) \in \Delta$ exactly intertwines $P(s)$ with $P(t)$, that is,  
\begin{align} \label{eq: def adiab zeitentw} 
P(t) U(t,s) = U(t,s) P(s) 
\end{align}
for every $(s,t) \in \Delta$.
A simple 
way of obtaining 
adiabatic evolutions w.r.t.~some given family $P$ 
(independently observed by Kato in~\cite{Kato50} and Daleckii--Krein in~\cite{DaleckiiKrein50}) is described in the following proposition.

\begin{prop}[Kato, Daleckii--Krein] \label{prop: intertwining relation}
Suppose $A(t): D(A(t)) \subset X \to X$ for every $t \in I$ is a densely defined closed linear operator and $P(t)$ a bounded projection in $X$ such that $P(t)A(t) \subset A(t)P(t)$ for every $t \in I$ and $t \mapsto P(t)$ is strongly continuously differentiable. If the evolution system $V_{\eps}$ for $\frac 1 \eps A + [P',P]$ exists on $D(A(t))$ for every $\eps \in (0,\infty)$, then $V_{\eps}$ is adiabatic w.r.t.~$P$ for every $\eps \in (0,\infty)$.
\end{prop}

\begin{proof}
Choose an arbitrary $(s,t) \in \Delta$ with $s \ne t$. 
It then follows by Lemma~\ref{lm: zeitentw rechtsseit db} that, for every $x \in D(A(s))$, the map 
\begin{align*}
[s,t] \ni \tau \mapsto V_{\eps}(t,\tau) P(\tau) V_{\eps}(\tau,s)x
\end{align*}
is continuous and right differentiable. Since $P(\tau)$ commutes with $A(\tau)$ 
and 
\begin{align} \label{eq: PP'P=0}
P(\tau)P'(\tau)P(\tau) = 0 
\end{align}
for every $\tau \in I$ (which follows by applying $P$ from the left and the right to the identity $P' = (PP)' = P'P+PP'$), it further follows that the right derivative of this map is identically $0$ and so (by Corollary~2.1.2 of~\cite{Pazy}) this map is constant. 
In particular,
\begin{align*}
P(t)V_{\eps}(t,s)x - V_{\eps}(t,s)P(s)x = V_{\eps}(t,\tau) P(\tau) V_{\eps}(\tau,s)x \big|_{\tau=s}^{\tau=t} = 0,
\end{align*}
as desired. 
\end{proof}

\section{Adiabatic theorems for general linear operators} \label{sect:ad-thms-general}

After having provided the most important preliminaries, 
we can now establish our adiabatic theorems for general linear operators $A(t): D(A(t)) \subset X \to X$ with time-dependent domains. 
%
We thereby extend the respective adiabatic theorems for operators with time-independent domains from~\cite{zeitunabh} (Section~\ref{sect: adsätze mit sl, zeitabh} and~\ref{sect: adsätze ohne sl, zeitabh}) and~\cite{JoyePfister93} (Section~\ref{sect: höherer adsatz}). 
(Strictly speaking, 
the theorems in Section~\ref{sect: adsätze mit sl, zeitabh} and~\ref{sect: adsätze ohne sl, zeitabh} are generalizations only of  sligthly weakened versions of the theorems from~\cite{zeitunabh}, namely of the versions where all $W^{n,1}_*$-regularity assumptions are strengthened to $n$ times  strong continuous differentiability assumptions. See~\cite{Schmid15qmath} for such simplified versions.)
%
What changes compared to the adiabatic theorems from~\cite{zeitunabh} and~\cite{JoyePfister93} is, in essence, only the regularity assumptions: 
for instance, the regularity and stability condition on $t \mapsto A(t)$ of the theorems from~\cite{zeitunabh} will be replaced by strong continuous differentiability conditions on the resolvent map $t \mapsto (z-A(t))^{-1}$ for suitable $z \in \C$ and by the condition that the evolution for $\frac{1}{\eps} A$ exist on the spaces $D(A(t))$ and be bounded in $\eps \in (0,\infty)$.  
%
Also, the proofs from~\cite{zeitunabh} and~\cite{JoyePfister93} have to be changed 
only slightly, because most steps of those proofs -- in particular the crucial step from~\cite{zeitunabh} where the (approximate) commutator equation is resolved -- are pointwise in $t$. 

\begin{cond} \label{cond: U_eps existiert und beschränkt}
$A(t): D(A(t)) \subset X \to X$ for every $t \in I$ is a densely defined closed linear operator such that, for every $\eps \in (0,\infty)$, there is an evolution system $U_{\eps}$ for $\frac 1 \eps A$ on $D(A(t))$ and there is a constant $M \in [1,\infty)$ such that $\norm{U_{\eps}(t,s)} \le M$ for all $(s,t) \in \Delta$ and $\eps \in (0,\infty)$. 
\end{cond}

We point out that there is a large number of papers establishing the existence of evolution systems $U$ for a given family $A$ of linear operators $A(t)$ on $D(A(t))$ as, for instance, \cite{Kato56}, \cite{Kisynski63}, \cite{Tanabe60}, \cite{KatoTanabe62}, \cite{FujieTanabe73}, \cite{AcquistapaceTerreni87}.
See the survey article~\cite{Schnaubelt02} for many more references. 
%
Instead of working with evolution systems on the spaces $Y_t = D(A(t))$ as in Condition~\ref{cond: U_eps existiert und beschränkt}, one could also prove adiabatic theorems employing evolution systems for $A$ on certain subspaces $Y$ of the intersection of all $D(A(t))$ (as in~\cite{Kato70} or~\cite{Kato73}, for instance), 
but then one would have 
to impose various invariance conditions on the subspace $Y$, such as the $A(t)$-admissibiltity of $Y$, 
the invariance
\begin{align} \label{eq: sect 5, invarianz}
(z-A(t))^{-1} Y \subset Y
\end{align}
for $z \in \ran \gamma_t$ (case with spectral gap) or for $z \in \{ \lambda(t) + \delta e^{i \vartheta(t)}: \delta \in (0,\delta_0] \}$ (case without spectral gap),
and the invariance of $Y$ under $P(t)$ and $P'(t)$. 
Such invariance conditions, however, are difficult 
to verify in practice: the invariance \eqref{eq: sect 5, invarianz}, for instance, would be clear only for complex numbers $z$ with sufficiently large positive real part 
(Proposition~2.3 of~\cite{Kato70}).

\subsection{Adiabatic theorems with spectral gap condition}  \label{sect: adsätze mit sl, zeitabh}

We will need the following condition depending on $m \in \{ 0 \} \cup  \N \cup \{\infty \}$, 
the number of points at which $\sigma(\,.\,)$ falls into $\sigma(A(\,.\,)) \setminus \sigma(\,.\,)$. See~\cite{zeitunabh} (Section~2.5) for the definition of this notion and of the continuity of set-valued maps.  

\begin{cond} \label{cond: vor adsatz mit sl}
$A(t): D(A(t)) \subset X \to X$ for every $t \in I$ is a linear operator such that Condition~\ref{cond: U_eps existiert und beschränkt} is satisfied.
$\sigma(t)$ for every $t \in I$ is a compact 
subset of $\sigma(A(t))$, $\sigma(\,.\,)$ falls into $\sigma(A(\,.\,)) \setminus \sigma(\,.\,)$ at exactly $m$ points that accumulate at only finitely many points, and $I \setminus N \ni t \mapsto \sigma(t)$ is continuous, where $N$ denotes the set of those $m$ points at which $\sigma(\,.\,)$ falls into $\sigma(A(\,.\,)) \setminus \sigma(\,.\,)$. 
Also, 
\begin{gather*}
J_{t_0} \ni t \mapsto (z-A(t))^{-1}  \text{ is strongly continuously differentiable for all } z \in \ran \gamma_{t_0}, \\ \ran \gamma_{t_0} \ni z \mapsto \ddt{ (z-A(t))^{-1} } \text{ is strongly continuous for all } t \in J_{t_0}, \\
\sup_{ (t,z) \in J_{t_0} \times \ran \gamma_{t_0} } \norm{ \ddt{ (z-A(t))^{-1} } } < \infty
\end{gather*}
for every $t_0 \in I \setminus N$, where the cycle $\gamma_{t_0}$ and the non-trivial closed interval $J_{t_0} \ni t_0$ are such that $\ran \gamma_{t_0} \subset \rho(A(t))$ and $\operatorname{n}(\gamma_{t_0}, \sigma(t)) = 1$ and $\operatorname{n}(\gamma_{t_0}, \sigma(A(t)) \setminus \sigma(t)) = 0$ for every $t \in J_{t_0}$.
And finally, $P(t)$ is the 
projection associated with $A(t)$ and $\sigma(t)$ for every $t \in I \setminus N$ and $I \setminus N \ni t \mapsto P(t)$ extends to a twice strongly continuously differentiable map on the whole of $I$.
\end{cond}

With this condition at hand, we can now 
prove an adiabatic theorem with uniform spectral gap condition ($m = 0$) and non-uniform spectral gap condition ($m \in \N \cup \{\infty\}$) for operators $A(t)$ with time-dependent domains.

\begin{thm} \label{thm: adsatz mit sl, zeitabh}
Suppose $A(t)$, $\sigma(t)$, $P(t)$ for $t \in I$ are such that Condition~\ref{cond: vor adsatz mit sl} is satisfied with $m = 0$ or $m \in \N \cup \{ \infty\}$, respectively. Then
\begin{align*}
\sup_{t \in I} \norm{ U_{\eps}(t) - V_{\eps}(t) } = O(\eps)  \text{ resp. } o(1) \quad (\eps \searrow 0), 
\end{align*}
whenever the evolution system $V_{\eps}$ for $\frac 1 \eps A + [P',P]$ exists on $D(A(t))$ for all $\eps \in (0,\infty)$.
\end{thm}

\begin{proof}
We have to prove the theorem only in the case of a uniform spectral gap ($m = 0$) because in the case of a non-uniform spectral gap ($m \in \N \cup \{ \infty \}$) it can be reduced to the case $m=0$ in the same way as in~\cite{zeitunabh} (proof of Theorem~3.2). 
So, suppose that Condition~\ref{cond: vor adsatz mit sl} is satisfied with $m=0$. We can then argue in much 
the same way as in~\cite{zeitunabh} (proof of Theorem~3.1) to prove the claimed convergence. 
Indeed, we define the operators $B(t)$ as in~\cite{zeitunabh}, that is, 
\begin{align*}
 B(t)x := \frac{1}{2 \pi i} \int_{\gamma_{t_0}} (z-A(t))^{-1} P'(t) (z-A(t))^{-1} x \, dz
\end{align*}
for all $t \in J_{t_0}$, $t_0 \in I$ and $x \in X$, where $\gamma_{t_0}$ and $J_{t_0}$ are now given by Condition~\ref{cond: vor adsatz mit sl}. 
\smallskip

As a first preparatory step, we observe that the operators $B(t)$ satisfy the commutator equation
\begin{align} \label{eq: commutator equation}
 B(t) A(t) - A(t) B(t) \subset [P'(t),P(t)]
\end{align}
for all $t \in I$. 
As a second preparatory step, we observe that $t \mapsto B(t)$ is strongly continuously differentiable. 
Indeed, by Condition~\ref{cond: vor adsatz mit sl}, the map $J_{t_0} \ni t \mapsto (z-A(t))^{-1}$ is, in particular, norm continuous and therefore $J_{t_0} \ni t \mapsto A(t)$ is continuous in the generalized sense (Theorem~IV.2.25 of~\cite{KatoPerturbation80}) and
\begin{align} \label{eq:ad thm with sg, 1}
\sup_{ (t,z) \in J_{t_0} \times \ran \gamma_{t_0} } \norm{  (z-A(t))^{-1}  } < \infty 
\end{align}
(Theorem~IV.3.15 of~\cite{KatoPerturbation80}). 
Condition~\ref{cond: vor adsatz mit sl} and \eqref{eq:ad thm with sg, 1} now imply that the standard result for the differentiation of parameter-dependent (path) integrals is applicable and thus, by that result, $t \mapsto B(t)x$ is continuously differentiable for all $x \in X$, as claimed. 
\smallskip

With these preparations at hand, we can now proceed in almost literally the same way as in~\cite{zeitunabh}. 
%
Indeed, for $x \in D(A(0))$ the map $s \mapsto U_{\eps}(t,s) V_{\eps}(s)x$ is continuously differentiable (Lemma~\ref{lm: zeitentw rechtsseit db} and Corollary~2.1.2 of~\cite{Pazy}) and therefore we get, exploiting the commutator equation~\eqref{eq: commutator equation}, that
\begin{align} \label{eq:ad thm with sg, 2}
V_{\eps}(t)x - U_{\eps}(t)x &= U_{\eps}(t,s)V_{\eps}(s)x \big|_{s=0}^{s=t} 
= \int_0^t U_{\eps}(t,s) [P'(s),P(s)] V_{\eps}(s)x \, ds \notag \\
&= \int_0^t U_{\eps}(t,s) \bigl( B(s)A(s) - A(s)B(s) \bigr) V_{\eps}(s)x \, ds
\end{align}
for all $t \in I$. Additionally, for $x \in D(A(0))$ the map $s \mapsto U_{\eps}(t,s) B(s) V_{\eps}(s)x$ is continuously differentiable (by the strong continuous differentiability of $s \mapsto B(s)$ and by Lemma~\ref{lm: zeitentw rechtsseit db} and Corollary~2.1.2 of~\cite{Pazy}) and therefore we get from~\eqref{eq:ad thm with sg, 2} by partial integration that
\begin{align} \label{eq:ad thm with sg, 3}
&V_{\eps}(t)x - U_{\eps}(t)x = \eps \int_0^t U_{\eps}(t,s) \Bigl( - \, \frac 1 \eps A(s)B(s) + B(s) \frac 1 \eps A(s) \Bigr) V_{\eps}(s)x \, ds  \\
&\quad = \eps \, U_{\eps}(t,s) B(s) V_{\eps}(s)x \big|_{s=0}^{s=t} 
- \eps \int_0^t U_{\eps}(t,s) \bigl( B'(s) + B(s) [P'(s),P(s)] \bigr) V_{\eps}(s)x \, ds  \notag
\end{align}
for all $t \in I$ and $\eps \in (0,\infty)$. Since $U_{\eps}$ and $V_{\eps}$ are bounded above by an $\eps$-independent constant (Condition~\ref{cond: U_eps existiert und beschränkt} and Proposition~\ref{prop: störreihe für gestörte zeitentw}), the asserted convergence in the case $m=0$
immediately follows from~\eqref{eq:ad thm with sg, 3}. 
\end{proof}

In general, the existence of the evolution system $V_{\eps}$ for $\frac 1 \eps A + [P',P]$ on $D(A(t))$ does not seem to be guaranteed under the fairly general Condition~\ref{cond: vor adsatz mit sl}. 
(In view of Proposition~\ref{prop: störreihe für gestörte zeitentw} one would, of course, like to define $V_{\eps}$ as a perturbation series and show that $[s,1] \ni t \mapsto V_{\eps}(t,s)y$ for every $y \in D(A(s))$ is a continuously differentiable solution to the initial value problem $x' = \frac 1 \eps A(t)x + [P'(t),P(t)]x$ with $x(s) = y$, but this is not clear in general. See the remarks after Proposition~\ref{prop: störreihe für gestörte zeitentw}.)  
It is therefore good to know 
that under Condition~\ref{cond: vor adsatz mit sl} with $m=0$ 
one has at least the following statement:
\begin{align} \label{eq: zusatz zu adsatz mit sl, zeitabh}
\sup_{t \in I} \norm{ (1-P(t))U_{\eps}(t)P(0) }, \quad \sup_{t \in I} \norm{ P(t)U_{\eps}(t)(1-P(0)) } = O(\eps)  
\end{align}
as $\eps \searrow 0$, which follows from the adiabatic theorem of higher order (Theorem~\ref{thm: höherer adsatz}~(i) and~(iii) with degree of regularity $n = 1$) below. 
%
%
It should be pointed out, however, that Theorem~\ref{thm: adsatz mit sl, zeitabh} itself 
-- operating with the evolution systems for $\frac 1 \eps A + [P',P] = \frac{1}{\eps} A_{0 \, \eps} + K_{0 \, \eps} \ne \frac{1}{\eps} A_{1 \, \eps} + K_{1 \, \eps}$ as comparison evolutions --
is not contained in Theorem~\ref{thm: höherer adsatz}. 

\subsection{Adiabatic theorems without spectral gap condition}   \label{sect: adsätze ohne sl, zeitabh}

We now prove an adiabatic theorem without spectral gap condition for operators $A(t)$ with time-dependent domains. In contrast to the respective result from~\cite{zeitunabh} we have to explicitly require the differentiability of the resolvent as well as an estimate on the derivative of the resolvent because these two things are no longer automatically satisfied in the case of time-dependent domains. 

\begin{thm} \label{thm: erw adsatz ohne sl, zeitabh}
Suppose $A(t): D(A(t)) \subset X \to X$ for every $t \in I$ is a linear operator such that Condition~\ref{cond: U_eps existiert und beschränkt} is satisfied. Suppose further that $\lambda(t)$ for every $t \in I$ is an eigenvalue of $A(t)$, and that there are numbers $\delta_0 \in (0,\infty)$ and $\vartheta(t) \in \R$ such that $\lambda(t) + \delta e^{i \vartheta(t)} \in \rho(A(t))$ for all $\delta \in (0,\delta_0]$ and $t \in I$ and such that $t \mapsto \lambda(t)$ and $t \mapsto e^{i \vartheta(t)}$ are continuously differentiable and $t \mapsto \big( \lambda(t) + \delta e^{i \vartheta(t)} - A(t) \big)^{-1}$ is strongly continuously differentiable.
Suppose finally that $P(t)$ for every $t \in I$ is a bounded projection in $X$ 
such that $P(t)$ for almost every $t \in I$ is weakly associated with $A(t)$ and $\lambda(t)$ and that 
\begin{align} \label{eq:commut and incl ass ad thm wo sg}
P(t)A(t) \subset A(t)P(t) 
\qquad \text{and} \qquad
P(t)X \subset \ker (A(t)-\lambda(t))^{m_0}  
\end{align} 
for every $t \in I$ (and some $m_0 \in \N$). Additionally, suppose that there are $M_0, M_0' \in (0,\infty)$ such that 
\begin{align*}
&\norm{ \big( \lambda(t) + \delta e^{i \vartheta(t)} - A(t) \big)^{-1} (1-P(t)) } \le \frac{M_0}{\delta}, \\
& \qquad \qquad \qquad \qquad \norm{  \ddt{ \Big( \big( \lambda(t) + \delta e^{i \vartheta(t)} - A(t) \big)^{-1} } (1-P(t)) \Big)    } \le \frac{M_0'}{\delta^{m_0+1}}
\end{align*}
for all $\delta \in (0, \delta_0]$ and $t \in I$,
let $\rk P(0) < \infty$ and let $t \mapsto P(t)$ be strongly continuously differentiable. 
\begin{itemize}
\item[(i)] If $X$ is arbitrary (not necessarily reflexive), then 
\begin{align*} 
\sup_{t \in I} \norm{ \big( U_{\eps}(t) - V_{0\,\eps}(t) \big) P(0) } \longrightarrow 0 \quad (\eps \searrow 0),
\end{align*}
where $V_{0\,\eps}$ for $\eps \in (0,\infty)$ denotes the evolution system for $\frac 1 \eps  A P + [P',P]$ on $X$. 
\item[(ii)] If $X$ is reflexive and $t \mapsto P(t)$ is norm continuously differentiable, then  
\begin{align*}
\sup_{t \in I} \norm{ U_{\eps}(t) - V_{\eps}(t) } \longrightarrow 0 \quad (\eps \searrow 0),
\end{align*}
whenever the evolution system $V_{\eps}$ for $\frac 1 \eps A + [P',P]$ exists on $D(A(t))$ for $\eps \in (0, \infty)$.
\end{itemize}
\end{thm}

\begin{proof}
We have to modify the proof of the respective adiabatic theorem for operators with time-independent domains from~\cite{zeitunabh} (Theorem~4.2) only slightly. 
We begin with some preparations which are essential for the proof of both assertion~(i) and (ii). 
\smallskip

As a first preparatory step, we observe that the approximate commutator equation from~\cite{zeitunabh} can be solved in exactly the same way in the present case of time-dependent domains. 
Indeed, exactly as in~\cite{zeitunabh} we define the operators 
\begin{align}
&B_{n \, \bm{\delta}}(t) := \sum_{k=0}^{m_0-1} \Big(  \prod_{i=1}^{k+1} \ol{R}_{\delta_i}(t) \Big) Q_n(t) (\lambda(t)-A(t))^k P(t) \notag \\
&\qquad \qquad \qquad \qquad \qquad \quad + \sum_{k=0}^{m_0-1} (\lambda(t)-A(t))^k P(t) Q_n(t) \Big(  \prod_{i=1}^{k+1} \ol{R}_{\delta_i}(t) \Big)
\end{align}
for $n \in \N$, $\bm{\delta} := (\delta_1, \dots, \delta_{m_0}) \in (0,\delta_0]^{m_0}$ and $t \in I$, where
\begin{align*}
\ol{R}_{\delta}(t) := R_{\delta}(t) \ol{P}(t) \quad \text{with} \quad R_{\delta}(t) := \big( \lambda(t) + \delta e^{i \vartheta(t)} - A(t) \big)^{-1} \text{\, and \,\,} \ol{P}(t) := 1-P(t)
\end{align*}
for $\delta \in (0,\delta_0]$, and where 
\begin{align*}
Q_n(t) := \int_0^1 J_{1/n}(t-r) P'(r) \, dr
\end{align*}
with $(J_{1/n})$ being a standard mollifier in $C_c^{\infty}((0,1),\R)$. With the same calculations as in~\cite{zeitunabh} we then get that the operators $B_{n \, \bm{\delta}}(t)$ satisfy the approximate commutator equation
\begin{align} \label{eq: appr commutator equation}
B_{n \, \bm{\delta}}(t)A(t) - A(t)B_{n \, \bm{\delta}}(t) + C_{n \, \bm{\delta}}(t) \subset [Q_n(t),P(t)]
\end{align}
with remainder terms $C_{n \, \bm{\delta}}(t) = C_{n \, \bm{\delta}}^+(t) - C_{n \, \bm{\delta}}^-(t)$ which are given by
\begin{align}
&C_{n \, \bm{\delta}}^+(t) := \sum_{k=0}^{m_0-1} \delta_{k+1} e^{i \vartheta(t)} \Big( \prod_{i=1}^{k+1} \ol{R}_{\delta_i}(t) \Big) Q_n(t) (\lambda(t)-A(t))^k P(t),  \notag \\
&\qquad \qquad \qquad C_{n \, \bm{\delta}}^-(t) := \sum_{k=0}^{m_0-1} (\lambda(t)-A(t))^k P(t) Q_n(t) \, \delta_{k+1} e^{i \vartheta(t)} \Big( \prod_{i=1}^{k+1} \ol{R}_{\delta_i}(t) \Big)
\end{align}
and which are suitably controlled later on in the proof. 
In the calculations leading to~\eqref{eq: appr commutator equation}, the relations~\eqref{eq:commut and incl ass ad thm wo sg} are essential -- just like in~\cite{zeitunabh}. (Yet, in~\cite{zeitunabh} these relations did not have to be assumed but could be inferred from the weak associatedness of $P(t)$ with $A(t)$ and $\lambda(t)$ for almost every $t$, the finiteness of $\rk P(t)$, and the strong continuous differentiability of $t \mapsto P(t)$. See the first prepartory step of the proof of Theorem~4.2 from~\cite{zeitunabh}).
\smallskip

As a second preparatory step, we observe that $t \mapsto B_{n \, \bm{\delta}}(t)$ is strongly continuously differentiable and that $B_{n \, \bm{\delta}}$ and $B_{n \, \bm{\delta}}'$ can be estimated appropriately.  
Indeed, since $$P(t)X \subset \ker(A(t)-\lambda(t))^{m_0} \subset D(A(t)^{m_0})$$ 
for every $t \in I$ by~(\ref{eq:commut and incl ass ad thm wo sg}), 
we see that $(A(t)-\lambda(t))P(t)$ is a bounded linear operator in $X$ for every $t \in I$ and that 
\begin{align}  \label{eq: stet von A P, lm 2 zum erw adsatz ohne sl}
t \mapsto \,\, &(A(t)-\lambda(t))P(t) = (A(t)-\lambda(t)) S_{\delta}(t) \,\, S_{\delta}(t)^{m_0-1} \big( A(t)-\lambda(t)-\delta e^{i \vartheta(t)} \big)^{m_0} P(t) \notag \\
&= \big( 1 + \delta e^{i \vartheta(t)} S_{\delta}(t) \big) \, \sum_{k=0}^{m_0-1} \binom{m_0}{k}  \big( -\delta e^{i \vartheta(t)} \big)^{m_0-k} \cdot \notag \\ 
&\qquad \qquad \qquad \qquad \qquad \qquad \qquad \cdot S_{\delta}(t)^{m_0-1 - k} \, \big( 1 + \delta e^{i \vartheta(t)} S_{\delta}(t) \big)^k P(t) 
\end{align}
is strongly continuously differentiable, where in addition to~(\ref{eq:commut and incl ass ad thm wo sg}.a) the binomial formula and the strong continuous differentiability assumption on $t \mapsto S_{\delta}(t) := \big( A(t)-\lambda(t)-\delta e^{i \vartheta(t)} \big)^{-1}$ have been used. 
So, 
\begin{align} \label{eq: 1, adsatz ohne sl, zeitabh}
t \mapsto (A(t)-\lambda(t))^k P(t) = ((A(t)-\lambda(t))P(t))^k
\end{align}
for $k \in \{1, \dots, m_0\}$ is strongly continuously differentiable as well and the desired strong continuous differentiability of $t \mapsto B_{n \, \bm{\delta}}(t)$ follows. 
It also follows from~\eqref{eq: 1, adsatz ohne sl, zeitabh} and the assumed resolvent estimates that 
the same (or even better) estimates for $B_{n \, \bm{\delta}}$ and $B_{n \, \bm{\delta}}'$ as in~\cite{zeitunabh} hold true, namely 
\begin{align}  \label{eq: absch B_n eps und B_n eps', zeitabh}
\sup_{t \in I} \big\| B_{n \, \bm{\delta}}(t) \big\| \le \sum_{k=1}^{m_0} c \, \Big( \prod_{i=1}^{k} \delta_i \Big)^{-1},
\qquad
\sup_{t \in I} \norm{ B_{n \, \bm{\delta}}'(t) } \le \sum_{k=1}^{m_0} c_n \, \Big( \prod_{i=1}^{k} \delta_i \Big)^{-(m_0+1)}
\end{align}
with positive constants $c, c_n$. 
\smallskip 

As a third and last preparatory step, we observe that for every $\eps \in (0,\infty)$ the evolution system $V_{0\,\eps}$ for $\frac 1 \eps AP + [P',P]$ exists on $X$ and is adiabatic w.r.t.~$P$ and satisfies the estimate
\begin{align} \label{eq: absch V_{0 eps}}
\norm{ V_{0\,\eps}(t,s) P(s) } \le Mc \, e^{Mc(t-s)} 
\end{align}
for all $(s,t) \in \Delta$, where $c$ is an upper bound of $t \mapsto \norm{P(t)}, \norm{ P'(t) }$. 
Indeed, since $t \mapsto A(t)P(t)$ is strongly continuous by the strong continuous differentiability of~\eqref{eq: stet von A P, lm 2 zum erw adsatz ohne sl}, we can argue in the same way as in the fourth preparatory step of the proof in~\cite{zeitunabh} -- just notice that the continuous differentiability of $\tau \mapsto U_{\eps}(t,\tau)V_{0\,\eps}(\tau,s)P(s)x$ now has to be concluded by Lemma~\ref{lm: zeitentw rechtsseit db} and Corollary~2.1.2 of~\cite{Pazy}. 
\smallskip 
 
With these preparations at hand, we can now proceed in almost literally the same way 
as in~\cite{zeitunabh} to prove the assertions~(i) and~(ii). In fact, apart from the 
preparatory steps, there is only two things that have to be changed, namely the justification of the various integration by parts steps and of the boundedness of $V_{\eps}$ with an $\eps$-independent bound.
Specifically, the fact that
\begin{align*}
[0,t] \ni s \mapsto U_{\eps}(t,s) B_{n \, \bm{\delta}}(s) V_{0\,\eps}(s) P(0)x \quad \text{or} \quad [0,t] \ni s \mapsto U_{\eps}(t,s) B_{n \, \bm{\delta}}(s) V_{\eps}(s) x
\end{align*}
is the continuous representative of an element of the Sobolev space  $W^{1,1}([0,t],X)$ 
for all $x \in X$ or~all $x \in D(A(0))$, respectively, 
can no longer be deduced from the $W^{1,1}_*$-product rule (Lemma~2.2 of~\cite{zeitunabh}) but has to be inferred from Corollary~2.1.2 of~\cite{Pazy} using that $s \mapsto B_{n \, \bm{\delta}}(s)$ is strongly continuously differentiable with $B_{n \, \bm{\delta}}(s) X \subset D(A(s))$ for every $s \in I$. 
And, 
the fact that $V_{\eps}$ is bounded by an $\eps$-independent bound can no longer be deduced from the standard perturbation result for $(M,\omega)$-stability, but one has to invoke Proposition~\ref{prop: störreihe für gestörte zeitentw}.
\end{proof}


In a similar manner, one sees that the 
variants of the qualitative adiabatic theorem without spectral gap condition from~\cite{zeitunabh} 
carry over to the case of time-dependent domains as well, provided their hypotheses are adapted in a similar way as above. In particular,  this is true for the quantitative variants and the variants tailored to spectral operators (Corollary~4.3, Theorem~4.4 and Corollary~4.5 from~\cite{zeitunabh}).

\subsection{An adiabatic theorem of higher order}  \label{sect: höherer adsatz}

In this section we extend the adiabatic theorem of higher order of Joye and Pfister from~\cite{JoyePfister93} to the case of general operators $A(t)$ with possibly time-dependent domains. 
We will use the elegant iterative scheme of~\cite{JoyePfister93} which we briefly recall here (in a slightly modified form) for the reader's convenience. 
\smallskip

Suppose $A(t): D(A(t)) \subset X \to X$ is a densely defined closed linear operator and $\gamma_t$ is a cycle in $\C$ for every $t \in J$, where $J$ is a compact interval, and let $\eps \in (0, \infty)$ and $n \in \N$. Then $A_{0 \, \eps}$, $P_{0 \, \eps}$, $K_{0 \, \eps}$ are called \emph{well-defined w.r.t.~$\gamma_t$ ($t \in J$)} if and only if $\ran \gamma_t \subset \rho(A_{0 \, \eps}(t))$ for all $t \in J$, where $A_{0 \, \eps}(t) := A(t)$, and $J \ni t \mapsto P_{0 \, \eps}(t)$ is strongly continuously differentiable, where
\begin{align*} 
P_{0 \, \eps}(t):= \frac{1}{2 \pi i} \int_{\gamma_t} (z-A_{0 \, \eps}(t))^{-1} \, dz. 
\end{align*}
In this case $K_{0 \, \eps}$ is defined by $K_{0 \, \eps}(t) := [P_{0 \, \eps}'(t), P_{0 \, \eps}(t)]$.
And, for general $n \in \N$, $A_{n \, \eps}$, $P_{n \, \eps}$, $K_{n \, \eps}$ are called \emph{well-defined w.r.t.~$\gamma_t$ ($t \in J$)} if and only if  $A_{n-1 \, \, \eps}$, $P_{n-1  \, \eps}$, $K_{n-1 \, \eps}$ are well-defined w.r.t.~$\gamma_t$ ($t \in J$),  $\ran \gamma_t \subset \rho(A_{n \, \eps}(t))$ for all $t \in J$, where $A_{n \, \eps}(t) := A(t) - \eps K_{n-1 \, \eps}(t)$, and $J \ni t \mapsto P_{n \, \eps}(t)$ is strongly continuously differentiable, where 
\begin{align*}
P_{n \, \eps}(t):= \frac{1}{2 \pi i} \int_{\gamma_t} (z-A_{n \, \eps}(t))^{-1} \, dz.
\end{align*}
In this case $K_{n \, \eps}$ is defined by $K_{n \, \eps}(t) := [P_{n \, \eps}'(t), P_{n \, \eps}(t)]$.
\smallskip

We will need the following conditions depending on $n \in \N \cup \{ \infty \}$, the degree of regularity, in the adiabatic theorem of higher order below. 

\begin{cond} \label{cond: vor höherer adsatz}
$A(t): D(A(t)) \subset X \to X$ for every $t \in I$ is a densely defined closed linear operator. 
$\sigma(t)$ for every $t \in I$ is a compact and isolated subset of $\sigma(A(t))$, there is an $r_0 > 0$ such that $B_{r_0}(\sigma(t)) \setminus \sigma(t) \subset \rho(A(t))$ for all $t \in I$, 
and $t \mapsto \sigma(t)$ is continuous. 
Also, for every $t_0 \in I$, there are positive constants $a_{t_0}$, $b_{t_0}$, $c_{t_0}$ such that 
\begin{gather*}
J_{t_0} \ni t \mapsto (z-A(t))^{-1}  \text{ is $n$ times strongly continuously differentiable for all } z \in \ran \gamma_{t_0}, \\ 
\ran \gamma_{t_0} \ni z \mapsto \ddtl{ (z-A(t))^{-1} } \text{ is strongly continuous for all } t \in J_{t_0}, l \in \{ 1, \dots, n \}, \\
\sup_{ (t,z) \in J_{t_0} \times \ran \gamma_{t_0} } \norm{ \ddtl{ (z-A(t))^{-1} } } \le a_{t_0} c_{t_0}^l \frac{l!}{(1+l)^2} \text{ for all } l \in \{ 1, \dots, n \},
\end{gather*}
where $\gamma_{t_0}$ is a cycle in $\ol{B}_{\frac{4 r_0}{7}}(\sigma(t_0)) \setminus B_{ \frac{3 r_0}{7} }(\sigma(t_0))$ with 
\begin{align*}
\operatorname{n}\big( \gamma_{t_0}, B_{ \frac{3 r_0}{7} }(\sigma(t_0)) \big) = 1 \quad \text{and} \quad \operatorname{n}\big( \gamma_{t_0}, \C \setminus \ol{B}_{ \frac{4 r_0}{7} }(\sigma(t_0)) \big) = 0
\end{align*}
and where $J_{t_0} \subset I$ is a non-trivial closed interval containing $t_0$ such that $\sigma(t) \subset B_{ \frac{r_0}{7} }(\sigma(t_0))$ and $\sigma(t_0) \subset B_{ \frac{r_0}{7} }(\sigma(t))$ for all $t \in J_{t_0}$.
And finally, $P(t)$ for every $t \in I$ is the 
projection associated with $A(t)$ and $\sigma(t)$,
$t \mapsto P(t)$ is $n+1$ times strongly continuously differentiable and
\begin{align*}
\sup_{t \in J_{t_0}} \norm{ \ddtl{ [P'(t),P(t)] } } \le b_{t_0} c_{t_0}^l \frac{l!}{(1+l)^2} \text{ for all } l \in \{ 0, 1, \dots, n \} \text{ and } t_0 \in I.
\end{align*}
\end{cond}

In the special case of time-independent domains $D(A(t)) = D$, one easily sees 
that the requirements on the resolvent of $A(t)$ in Condition~\ref{cond: vor höherer adsatz} are fulfilled for an $n \in \N$ if, for instance, $t \mapsto A(t)x$ is $n$ times strongly continuously differentiable for all $x \in D$. 
And they 
are fulfilled for $n = \infty$ if, for instance, there is an open neighbourhood $U_I$ of $I$ in $\C$ such that, for every $x \in D$, $t \mapsto A(t)x$ extends to a holomorphic map on $U_I$ (Cauchy inequalities!).

\begin{lm}[Joye--Pfister] \label{lm: iteration wohldef}
\begin{itemize}
\item[(i)] Suppose that Condition~\ref{cond: vor höherer adsatz} is satisfied for a finite $n \in \N$. Then there is an $\eps^* >0$ such that $A_{n \, \eps}$, $P_{n \, \eps}$, $K_{n \, \eps}$ are well-defined w.r.t.~$\gamma_t$ ($t \in I$) for every $\eps \in (0, \eps^*]$. Additionally, 
\begin{align*}
\sup_{t \in I} \norm{ K_{n \, \eps}(t) - K_{n-1 \, \eps}(t)} = O(\eps^n) \quad (\eps \searrow 0).
\end{align*}
\item[(ii)] Suppose that Condition~\ref{cond: vor höherer adsatz} is satisfied for $n = \infty$. Then there is an $\eps^* >0$ and for every $\eps \in (0, \eps^*]$ there is a natural number $n^*(\eps) \in \N$ such that $A_{n^*(\eps) \, \eps}$, $P_{n^*(\eps) \, \eps}$, $K_{n^*(\eps) \, \eps}$ are well-defined w.r.t.~$\gamma_t$ ($t \in I$) for every $\eps \in (0, \eps^*]$. Additionally, there is 
a constant $g \in (0, \infty)$ such that
\begin{align*}
\sup_{t \in I} \norm{ K_{n^*(\eps) \, \eps}(t) - K_{n^*(\eps)-1 \, \eps}(t)} = O\bigl( e^{-\frac{g}{\eps}} \bigr) \quad (\eps \searrow 0).
\end{align*}
\end{itemize}
\end{lm}

\begin{proof}
We begin with some general preparatory considerations from which both part~(i) and part~(ii) will easily follow. 
Suppose (for the entire proof) that Condition~\ref{cond: vor höherer adsatz} is satisfied for $n = 1$ 
and fix $t_0 \in I$ for the moment. 
We have 
\begin{align*}
J_{t_0} \times C_{t_0} := J_{t_0} \times \ol{B}_{ \frac{5 r_0}{7} }(\sigma(t_0)) \setminus B_{ \frac{2 r_0}{7} }(\sigma(t_0)) \subset \subset \bigl\{ (t,z) \in J_{t_0} \times \C: z \in \rho(A(t)) \bigr \} =: U_{t_0}
\end{align*}
and $U_{t_0} \ni (t,z) \mapsto (z-A(t))^{-1}$ is continuous (Theorem~IV.3.15 of~\cite{KatoPerturbation80}), because $J_{t_0} \ni t \mapsto A(t)$ is continuous in the generalized sense (Theorem~IV.2.25  of~\cite{KatoPerturbation80}) due to the strong continuous differentiability of $J_{t_0} \ni t \mapsto (z-A(t))^{-1}$. Consequently, $J_{t_0} \times C_{t_0} \ni (t,z) \mapsto (z-A(t))^{-1}$ is bounded, 
whence we can (and will) assume w.l.o.g. that 
\begin{align} \label{eq: w.l.o.g. assumption}
\sup_{(t,z) \in J_{t_0} \times C_{t_0}} \norm{ (z-A(t))^{-1} } \le a_{t_0}.
\end{align}
We now define $\eps_{t_0}^*$ and $n_{t_0}^*(\eps)$ just like in Joye and Pfister's paper~\cite{JoyePfister93}, that is, 
\begin{align} \label{eq: def von eps^*}
\nonumber \eps_{t_0}^* &:= \max \Bigl \{ \eps \in \bigl( 0, \frac{1}{2 a_{t_0} b_{t_0}} \bigr): \sum_{k=1}^{\infty} \bigl( 2 \alpha^2 a_{t_0} b_{t_0} \frac{\eps}{1-2 a_{t_0} b_{t_0} \eps} \bigr)^k \le \alpha \Bigr \}, \\
& \qquad \qquad \quad n_{t_0}^*(\eps) := \Big \lfloor \frac{1}{ e c_{t_0} d_{t_0} \, \eps }  \Big \rfloor \text{ for } \eps \in (0,\infty),
\end{align}
where $\alpha$ and $d_{t_0}$ are defined by equation~(2.30) and equation~(2.50) of~\cite{JoyePfister93}. (In particular,  $\eps_{t_0}^*$ and $n_{t_0}^*(\eps)$ only depend on $\gamma_{t_0}$, $a_{t_0}$, $b_{t_0}$ and $c_{t_0}$.) 
We now show by finite induction over $k$: whenever Condition~\ref{cond: vor höherer adsatz} is satisfied for a certain $n' \in \N$, then the following holds true for all $\eps \in (0, \eps_{t_0}^*]$ and all $k \in \{ 1, \dots, n_{t_0}^*(\eps, n') \}$ with $n_{t_0}^*(\eps, n') := \min\{ n_{t_0}^*(\eps), n' \}$: 
\begin{itemize}
\item[(a)] $A_{k \, \eps}$, $P_{k \, \eps}$, $K_{k \, \eps}$ are well-defined w.r.t.~$\gamma_t$ ($t \in J_{t_0}$) 
and $J_{t_0} \ni t \mapsto K_{k \, \eps}(t)$ is $n_{t_0}^*(\eps, n')-k$ times strongly continuously differentiable
\item[(b)] $\sup_{t \in J_{t_0}} \norm{ K_{k \, \eps}^{(l)}(t) - K_{k-1 \, \eps}^{(l)}(t) } \le b_{t_0} c_{t_0}^{k+l} d_{t_0}^k \eps^k \frac{(k+l)!}{(1+l)^2}$ for all $l \in \N \cup \{ 0 \}$ 
with the property that  
$k + l \le n_{t_0}^*(\eps, n')$
\item[(c)] $\sup_{t \in J_{t_0}} \norm{ K_{k \, \eps}^{(l)}(t) } \le 2 b_{t_0} c_{t_0}^l \frac{l!}{(1+l)^2}$ for all $l \in \N \cup \{ 0 \}$ with $k + l \le n_{t_0}^*(\eps, n')$.
\end{itemize}
Suppose that Condition~\ref{cond: vor höherer adsatz} is satisfied for a certain $n' \in \N$ and fix $\eps \in (0, \eps_{t_0}^*]$.
Set $k = 1$ for the induction basis. We have only to prove assertion~(a) since 
assertions~(b) and~(c) can be gathered from the proof of Proposition~2.1 of~\cite{JoyePfister93}. It is obvious that $A_{0 \, \eps}$, $P_{0 \, \eps}$, $K_{0 \, \eps}$ are well-defined w.r.t.~$\gamma_t$ ($t \in J_{t_0}$) and that $t \mapsto K_{0 \, \eps}(t) = [P'(t), P(t)]$ is $n'$ times strongly continuously differentiable. 
Since, for $z \in C_{t_0}$ and $t \in J_{t_0}$, 
\begin{gather*}
(z-A_{1 \, \eps}(t)) = \bigl( 1+ \eps K_{0 \, \eps}(t) (z-A(t))^{-1} \bigr) (z-A(t)) \\
\text{ and } \norm{ \eps K_{0 \, \eps}(t) (z-A(t))^{-1} } \le \eps b_{t_0} \norm{ (z-A(t))^{-1} } \le \eps_{t_0}^* b_{t_0} a_{t_0} < \frac{1}{2}
\end{gather*}
(remember the estimate for $K_{0 \, \eps} = [P',P]$ from Condition~\ref{cond: vor höherer adsatz}, the estimate for the resolvent of $A$ from~\eqref{eq: w.l.o.g. assumption}, and the definition of $\eps_{t_0}^*$ in~\eqref{eq: def von eps^*}), we see that 
\begin{align*}
\ran \gamma_t \subset \ol{B}_{ \frac{4 r_0}{7} }(\sigma(t)) \setminus B_{ \frac{3 r_0}{7} }(\sigma(t)) \subset \ol{B}_{ \frac{5 r_0}{7} }(\sigma(t_0)) \setminus B_{ \frac{2 r_0}{7} }(\sigma(t_0)) = C_{t_0} \subset \rho(A_{1 \, \eps}(t))
\end{align*}
for all $t \in J_{t_0}$. 
And since 
\begin{align*}
\operatorname{n}\big( \gamma_t, B_{ \frac{2 r_0}{7} }(\sigma(t_0)) \big) = 
1 &= \operatorname{n}\big( \gamma_{t_0}, B_{ \frac{2 r_0}{7} }(\sigma(t_0)) \big), \\
\operatorname{n}\big( \gamma_t, \C \setminus \ol{B}_{ \frac{5 r_0}{7} }(\sigma(t_0)) \big) = 
0 &= \operatorname{n}\big( \gamma_{t_0}, \C \setminus \ol{B}_{ \frac{5 r_0}{7} }(\sigma(t_0)) \big)
\end{align*}
and $C_{t_0} \subset \rho(A_{1 \, \eps}(t))$ for all $t \in J_{t_0}$, the cycles $\gamma_t$ and $\gamma_{t_0}$ are homologous in $\rho(A_{1 \, \eps}(t))$ for $t \in J_{t_0}$, so that
\begin{align*}
J_{t_0} \ni t \mapsto P_{1 \, \eps}(t) &= \frac{1}{2 \pi i} \int_{\gamma_{t}} (z-A_{1 \, \eps}(t))^{-1} \, dz \\
&= \frac{1}{2 \pi i} \int_{\gamma_{t_0}} (z-A(t))^{-1} \big( 1 + \eps K_{0 \, \eps}(t) (z-A(t))^{-1} \big)^{-1} \, dz
\end{align*}
is $n'$ times strongly continuously differentiable. 
(In order to see this, use the product rule and inverses rule for strong continuous differentiability 
as well as Condition~\ref{cond: vor höherer adsatz}.)
Consequently, $A_{1 \, \eps}$, $P_{1 \, \eps}$, $K_{1 \, \eps}$ are well-defined w.r.t.~$\gamma_t$ ($t \in J_{t_0})$ and $t \mapsto K_{1 \, \eps}(t)$ is $n'-1$ times (in particular, $n_{t_0}^*(\eps, n')-1$ times) strongly continuously differentiable.
\smallskip

Choose now $k \in \{2, \dots, n_{t_0}^*(\eps, n') \}$ and assume that assertions (a), (b), (c) are true for $k-1$. We then have  to show that they are also true for $k$. As above we have only to establish (a) since (b) and (c) can then be derived as in the proof of Proposition~2.1 of~\cite{JoyePfister93}, as a close inspection of that proof shows. And in order to prove (a) we can proceed essentially as above: just use assertion~(c) for $k-1$ 
to get the estimate 
\begin{align*}
\sup_{(t,z) \in J_{t_0} \times C_{t_0}} \norm{ \eps K_{k-1 \, \eps}(t) (z-A(t))^{-1} } \le 2 b_{t_0} a_{t_0} \eps_{t_0}^* < 1
\end{align*}
and continue as above, thereby concluding the inductive proof of~(a), (b), (c).
\smallskip

Choosing finitely many points $t_1, \dots, t_m \in I$ such that $J_{t_1} \cup \dots \cup J_{t_m} = I$, and setting 
\begin{align} \label{eq: def eps^* und n^*(eps), global}
\eps^* := \min \{ \eps_{t_1}^*, \dots, \eps_{t_m}^* \} \text{ and } n^*(\eps) := \min \{ n_{t_1}^*(\eps), \dots, n_{t_m}^*(\eps) \},
\end{align} 
we find -- in virtue of the above preparations -- that, for every $\eps \in (0, \eps^*]$, the following holds true: whenever Condition~\ref{cond: vor höherer adsatz} is fulfilled for an $n' \in \N$, then $A_{k \, \eps}$, $P_{k \, \eps}$, $K_{k \, \eps}$ are well-defined w.r.t.~$\gamma_t$ ($t \in I$) and
\begin{align} \label{eq: diff der K-terme}
\sup_{t \in I} \norm{K_{k \, \eps}(t) - K_{k-1 \, \eps}(t) } \le b c^k d^k \eps^k k!
\end{align}
for every $k \in \{ 1, \dots, n^*(\eps, n') \}$, where $b$, $c$, $d$ are obtained by taking the maximum of the corresponding quantities for the points $t_1, \dots, t_m$ and $n^*(\eps, n') := \min \{ n^*(\eps), n' \}$.
\smallskip

Suppose now as in~(i) that Condition~\ref{cond: vor höherer adsatz} is satisfied for an $n \in \N$. Since $n^*(\eps) \longrightarrow \infty$ as $\eps \searrow 0$, 
we can assume w.l.o.g. that $n^*(\eps,n) = n$ for all $\eps \in (0, \eps^*]$ 
and therefore assertion~(i) follows from~\eqref{eq: diff der K-terme}.
Suppose finally as in~(ii) that Condition~\ref{cond: vor höherer adsatz} is satisfied for $n = \infty$. Since for every $\eps \in (0, \eps^*]$ 
there is $n' \in \N$ such that $n^*(\eps, n') = n^*(\eps)$ and since Condition~\ref{cond: vor höherer adsatz} is satisfied, in particular, for this $n'$, assertion~(ii) follows from~\eqref{eq: diff der K-terme} 
with the help of Stirling's formula (see, for instance, the proof of Theorem~2.1 of~\cite{JoyePfister93} or of Theorem~1b of~\cite{Nenciu93}).
\end{proof}

After these preparations we can now prove the announced adiabatic theorem of higher order. It extends Theorem~2.1 of~\cite{JoyePfister93} where skew-adjoint operators $A(t)$ are considered that analytically depend on $t$ and have time-independent domains.

\begin{thm} \label{thm: höherer adsatz}
Suppose $A(t)$, $\sigma(t)$, $P(t)$ for $t \in I$ are such that Condition~\ref{cond: U_eps existiert und beschränkt} is satisfied and Condition~\ref{cond: vor höherer adsatz} 
is satisfied for a finite $n \in \N$ or for $n = \infty$, respectively. Then 
\begin{itemize}
\item[(i)] 
$$\sup_{t \in I} \norm{ P_{\eps}(t) - P(t) } = O(\eps) \quad (\eps \searrow 0), $$
where for all $\eps \in (0, \eps^*]$ and $t \in I$, $P_{\eps}(t) := P_{n \, \eps}(t)$ in case $n \in \N$ and $P_{\eps}(t) := P_{n^*(\eps) \, \eps}(t)$ in case $n = \infty$ and where $\eps^*$ and $n^*(\eps)$ are defined as in~\eqref{eq: def eps^* und n^*(eps), global} of the lemma above. 
\item[(ii)] Whenever 
the evolution system $V_{\eps}$ for $\frac{1}{\eps} A_{n \, \eps} + K_{n \, \eps}$ resp.~$\frac{1}{\eps} A_{n^*(\eps) \, \eps} +  K_{n^*(\eps) \, \eps}$ exists on $D(A(t))$ for all $\eps \in (0,\eps^*]$, then $V_{\eps}$ is adiabatic w.r.t.~$P_{\eps}$ and for a suitable constant $g \in (0,\infty)$
\begin{align*}
\sup_{t \in I} \norm{ V_{\eps}(t) - U_{\eps}(t) } = O(\eps^n)  \text{ resp. } O\bigl( e^{-\frac{g}{\eps}} \bigr) \quad (\eps \searrow 0).
\end{align*}
\item[(iii)] 
Additionally, one has -- the existence of $V_{ \eps}$ being  irrelevant here -- 
that 
\begin{align*}
&\sup_{t \in I} \norm{ (1-P_{\eps}(t)) U_{\eps}(t) P_{\eps}(0) }, \\
& \qquad \qquad \sup_{t \in I} \norm{ P_{\eps}(t)U_{\eps}(t)(1-P_{\eps}(0)) } = O(\eps^n)  \text{ resp. } O\bigl( e^{-\frac{g}{\eps}} \bigr) \quad (\eps \searrow 0).
\end{align*}
\end{itemize}
\end{thm}

\begin{proof}
(i) Set $A_{\eps}(t) := A_{n \, \eps}(t)$ and $K_{\eps}^-(t) := K_{n-1 \, \eps}(t)$ in case  $n \in \N$ and $A_{\eps}(t) := A_{n^*(\eps) \, \eps}(t)$ and $K_{\eps}^-(t) := K_{n^*(\eps)-1 \, \eps}(t)$ in case $n = \infty$ (for $t \in I$ and $\eps \in (0,\eps^*]$).
As was shown in the proof of the above lemma, the cycles $\gamma_{t}$ and $\gamma_{t_i}$ are homologous in $\rho(A_{\eps}(t))$ for every $t \in J_{t_i}$ (where $t_1, \dots, t_m$ are points of $I$ chosen as in the definition of $\eps^*$ and $n^*(\eps)$ in~\eqref{eq: def eps^* und n^*(eps), global}) and every $\eps \in (0,\eps^*]$, whence 
\begin{align} \label{eq:hoeherer adsatz (i)}
P_{\eps}(t) - P(t) &= \frac{1}{2 \pi i} \int_{\gamma_{t_i}} (z-A_{\eps}(t))^{-1} - (z-A(t))^{-1} \, dz \notag \\
&= - \frac{1}{2 \pi i} \int_{\gamma_{t_i}} (z-A_{\eps}(t))^{-1} \, \eps K_{\eps}^-(t) \, (z-A(t))^{-1} \, dz
\end{align} 
for all $t \in J_{t_i}$ and $\eps \in (0,\eps^*]$. 
Also, it was shown in the proof of the above lemma 
that for all $\eps \in (0, \eps^*]$ and all $i \in \{1, \dots, m\}$ one has $\sup_{(t,z) \in J_{t_i} \times \ran \gamma_{t_i}} \norm{ (z-A(t))^{-1} } \le a_{t_i}$, $\sup_{t \in J_{t_i}} \norm{ K_{\eps}^-(t) } \le 2 b_{t_i}$, and 
\begin{align*}
\norm{ (z-A_{\eps}(t))^{-1} } &\le  \norm{ (z-A(t))^{-1} } \norm{ \big( 1 + \eps K_{\eps}^-(t) (z-A(t))^{-1} \big)^{-1} } \\
&\le a_{t_i} \sum_{m=0}^{\infty} ( \eps 2 b_{t_i} a_{t_i} )^m \le \frac{a_{t_i}}{1-2 a_{t_i} b_{t_i} \eps_{t_i}^*} < \infty 
\end{align*}   
for all $(t,z) \in J_{t_i} \times \ran \gamma_{t_i}$. Inserting these three estimates into~\eqref{eq:hoeherer adsatz (i)}, we obtain assertion~(i). 
\smallskip

(ii) Set $K_{\eps}^+(t) := [P_{\eps}'(t), P_{\eps}(t)]$ for $t \in I$ and $\eps \in (0,\eps^*]$ 
and suppose that the evolution system $V_{\eps}$ for $$\frac{1}{\eps} A_{\eps} + K_{\eps}^+ = \frac{1}{\eps} A + K_{\eps}^+-K_{\eps}^-$$ exists on $D(A(t))$. 
Since $U_{\eps}$ and $K_{\eps}^+-K_{\eps}^-$ are bounded uniformly in $\eps \in (0,\eps^*]$ (Condition~\ref{cond: U_eps existiert und beschränkt} and Lemma~\ref{lm: iteration wohldef}), $V_{\eps}$ is bounded uniformly in $\eps \in (0,\eps^*]$ as well (Proposition~\ref{prop: störreihe für gestörte zeitentw}). 
Since, moreover, for every $x \in D(A(0))$ the map $[0,t] \ni s \mapsto U_{\eps}(t,s) V_{\eps}(s) x$ is continuous and right differentiable (Lemma~\ref{lm: zeitentw rechtsseit db}) and since the right derivative $s \mapsto U_{\eps}(t,s) \big( K_{\eps}^+(s)-K_{\eps}^-(s) \big) V_{\eps}(s) x$ is continuous (Lemma~\ref{lm: iteration wohldef}), 
it follows by Corollary~2.1.2 of~\cite{Pazy} that
\begin{align} \label{eq: höherer adsatz, intdarst}
V_{\eps}(t)x - U_{\eps}(t)x &= U_{\eps}(t,s) V_{\eps}(s) x \big |_{s=0}^{s=t} \nonumber \\
&= \int_0^t  U_{\eps}(t,s) \big( K_{\eps}^+(s)-K_{\eps}^-(s) \big) V_{\eps}(s) x \, ds
\end{align}
for all $t \in I$. Combining~\eqref{eq: höherer adsatz, intdarst} with the estimates from Lemma~\ref{lm: iteration wohldef} and with the uniform boundedness of $U_{\eps}$ and $V_{\eps}$ in $\eps \in (0,\eps^*]$, we obtain the claimed estimates. 
%
It remains to show that $V_{\eps}$ is adiabatic w.r.t.~$P_{\eps}$, but  this is an immediate consequence of Proposition~\ref{prop: intertwining relation}.   
%
\smallskip

(iii) Arguing as in the proof of~\eqref{eq: höherer adsatz, intdarst}  above, we see for every $x \in D(A(0))$ and every $t \in I$ that
\begin{align} \label{eq:hoeherer adsatz (iii), 1}
P_{\eps}(t) U_{\eps}(t)x - U_{\eps}(t)P_{\eps}(0)x &= U_{\eps}(t,s) P_{\eps}(s) U_{\eps}(s)x \big |_{s=0}^{s=t} \\
&= \int_0^t U_{\eps}(t,s) \Big( P_{\eps}'(s) - \frac{1}{\eps} \big( A(s) P_{\eps}(s) - P_{\eps}(s) A(s) \big) \Big) U_{\eps}(s)x \, ds. \notag
\end{align}
Since $A_{\eps}(s)$ commutes with $P_{\eps}(s)$ for $s \in I$ and since $A = A_{\eps} + \eps K_{\eps}^-$, we have
\begin{align} \label{eq:hoeherer adsatz (iii), 2}
P_{\eps}'(s) &- \frac{1}{\eps} \big( A(s)P_{\eps}(s) - P_{\eps}(s)A(s) \big) \subset P_{\eps}'(s) - [K_{\eps}^-(s), P_{\eps}(s)]  \\
&= P_{\eps}'(s) - [K_{\eps}^+(s), P_{\eps}(s)] + [K_{\eps}^+(s) - K_{\eps}^-(s), P_{\eps}(s)] = [K_{\eps}^+(s) - K_{\eps}^-(s), P_{\eps}(s)] \notag
\end{align}
for every $s \in I$. 
Combining~\eqref{eq:hoeherer adsatz (iii), 1} and~\eqref{eq:hoeherer adsatz (iii), 2} with the estimates from Lemma~\ref{lm: iteration wohldef} and with the uniform boundedness of $U_{\eps}$ and $P_{\eps}$ in $\eps \in (0,\eps^*]$ (Condition~\ref{cond: U_eps existiert und beschränkt} and part~(i)), we obtain assertion~(iii). 
\end{proof}

It is obvious from the definition of Joye and Pfister's iterative scheme that $P_{\eps}(t) = P(t)$ for all $t$ in the (possibly empty) set $I \setminus \supp P'$, and therefore it follows from Theorem~\ref{thm: höherer adsatz}~(iii) that, as $\eps \searrow 0$,  
\begin{align*}
\sup_{t \in I \setminus \supp P'} \norm{ (1-P(t))U_{\eps}(t)P(0) },  \qquad \sup_{t \in I \setminus \supp P'} \norm{ P(t)U_{\eps}(t)(1-P(0)) } 
\end{align*}
are of the orders $O(\eps^n)$ or $O\bigl( e^{-\frac{g}{\eps}} \bigr)$, respectively. See~\cite{AvronSeilerYaffe87}, \cite{JoyePfister93}, \cite{Nenciu93}, for instance, for analogs of this corollary. 
\smallskip

A result similar to Theorem~\ref{thm: höherer adsatz} could have been proved with the help of a method developed by Nenciu in~\cite{Nenciu93}. In fact, this can  easily be gathered from 
the exposition in Section~7 of~\cite{dipl}. We have chosen Joye and Pfister's method since it is 
easier to remember 
and effortlessly transferred to  
the case of several compact isolated subsets $\sigma_1(t), \dots, \sigma_m(t)$ of $\sigma(A(t))$ where each of them is uniformly isolated in $\sigma(A(t))$ and uniformly isolated from each of the others.
\smallskip

We finally comment on a recent 
superadiabatic-type theorem by Joye from~\cite{Joye07} dealing with time-independent domains and several spectral subsets $\sigma_i(t)$. It allows for a generalization 
of Condition~\ref{cond: U_eps existiert und beschränkt} at the cost of a specialization of Condition~\ref{cond: vor höherer adsatz}
and states the following (where we confine ourselves, for the sake of notational simplicity, to the case of only one spectral subset $\sigma_i(t) = \sigma(t)$):
if -- and what follows is a special case of Condition~\ref{cond: vor höherer adsatz} -- there is an open neighbourhood $U_I$ of $I$ such that $t \mapsto A(t)x$ for every $x \in D$ extends to a holomorphic map on $U_I$ and if $\sigma(t) = \{ \lambda(t) \}$ for every $t \in I$ for a uniformly isolated spectral value $\lambda(t)$ of $A(t)$ of finite algebraic multiplicity (hence an eigenvalue) such that 
$t \mapsto \lambda(t)$ is continuous, 
then it suffices for the conclusion of Theorem~\ref{thm: höherer adsatz} to hold that -- instead of 
Condition~\ref{cond: U_eps existiert und beschränkt} -- $\lambda(t)$ lie in the left closed complex half-plane and 
$A(t) \ol{P}(t)$ generate a contraction semigroup on $X$ for every $t \in I$ (where $\ol{P} := 1-P$).
So, in the above-mentioned special case of Condition~\ref{cond: vor höherer adsatz} 
the boundedness requirement on $U_{\eps}$ 
from Condition~\ref{cond: U_eps existiert und beschränkt} is not necessary for assertions~(i), (ii) and~(iii) of Theorem~\ref{thm: höherer adsatz}. 
It 
is, however, necessary for the convergences 
\begin{align*}
\sup_{t \in I} \norm{ (1-P(t))U_{\eps}(t)P(0) }, 
\quad \sup_{t \in I} \norm{ P(t)U_{\eps}(t)(1-P(0)) }  \longrightarrow 0  \quad (\eps \searrow 0) 
\end{align*}
with the originally given projections $P(t)$, 
which we are primarily interested in here. 
See the example at the end of Section~1 of~\cite{Joye07} for a proof of 
this necessity statement. 
Also, it should be remarked that the above-mentioned special requirements (analyticity and finite algebraic multiplicity) of Joye's theorem from~\cite{Joye07} are really essential for the proof in~\cite{Joye07}. 
Indeed, this proof essentially rests upon the following estimate for the evolution system $V_{0 \, \eps}$ for $\frac{1}{\eps} A_{0 \, \eps} + K_{0 \, \eps} = \frac{1}{\eps} A + [P',P]$ on $D$
\begin{align} \label{eq: entsch absch joye}
\sup_{(s,t) \in \Delta} \norm{ V_{0 \, \eps}(t,s) } \le c \, e^{c/ \eps^{\beta}} \quad (\eps \in (0,\eps^*])
\end{align}
with constants $\beta \in (0,1)$ and $c \in (0,\infty)$ (Proposition~6.1 of~\cite{Joye07}), which then -- by the usual perturbation argument (Proposition~\ref{prop: störreihe für gestörte zeitentw}) -- 
yields the estimates
\begin{align} \label{eq: absch 2 joye}
\sup_{(s,t) \in \Delta} \norm{ U_{\eps}(t,s) }, \, \, \sup_{(s,t) \in \Delta} \norm{ V_{\eps}(t,s) } \le c' \, e^{c'/ \eps^{\beta}} \quad (\eps \in (0,\eps^*])
\end{align}
from which, in turn, by the integral representation~\eqref{eq: höherer adsatz, intdarst} and the exponential decay of $K_{\eps}^+-K_{\eps}^-$ from Lemma~\ref{lm: iteration wohldef} (analyticity requirement!), the conclusion of Theorem~\ref{thm: höherer adsatz} finally follows. And the fundamental estimate~\eqref{eq: entsch absch joye}, in turn,  
rests upon  a result on the growth (in $\eps$) of the evolution system for analytic families $\frac{1}{\eps} N$ of nilpotent operators $N(t)$ on finite-dimensional spaces (Proposition~4.1 of~\cite{Joye07}), which proposition (by the analyticity and finite algebraic multiplicity requirement!) can be applied to the 
nilpotent endomorphisms 
\begin{align*}
N(t) := W(t)^{-1} (A(t)-\lambda(t)) W(t) \big|_{P(0)X}
\end{align*}
of the finite-dimensional space $P(0)X$ that analytically depend on $t$. In the equation above, $W$ denotes the evolution system for $[P',P]$ on $X$ and $W(t) := W(t,0)$, which exactly intertwines $P(0)$ and $P(t)$. 

\section{Adiabatic theorems for operators defined by symmetric sesquilinear forms} \label{sect:ad-thms-skew-adjoint}

In this section, we apply the general adiabatic theorems from the previous section to the special situation of skew-adjoint operators $A(t) = 1/i A_{a(t)}$ defined by densely defined closed symmetric sesquilinear forms $a(t)$ with time-independent (form) domain. In this special situation, the quite technical regularity assumptions of our general theorems can be easily ensured. 


\subsection{Some notation and preliminaries}

We start by recording 
the basic conditions (depending on a regularity parameter $n \in \N \cup \{\infty\}$) that shall be imposed on the sesquilinear forms $a(t)$ in the adiabatic theorems of this section.

\begin{cond} \label{cond: vor an a(t)}
$a(t): H^+ \times H^+ \to \C$ for every $t \in I$ is a symmetric 
sesquilinear form on the Hilbert space $H^+$ (with norm $\norm{\,.\,}^+$ and scalar product $\scprd{\,.\,,\,..\,}^+$) which is densely and continuously embedded into $H$ (with norm $\norm{\,.\,}$ and scalar product $\scprd{\,.\,,\,..\,}$). Also, there is a number $m \in (0,\infty)$ such that
\begin{align*}
\scprd{ \,.\, , \,..\, }_t^+ := a(t)(\,.\, , \,..\, ) + m \scprd{ \,.\, , \,..\, } 
\end{align*} 
is a scalar product on $H^+$ and such that the induced norm $\norm{\,.\,}_t^+$ is equivalent to $\norm{\,.\,}^+$ for every $t \in I$. 
And finally, $t \mapsto a(t)(x,y)$ is $n$ times continuously differentiable for all $x, y \in H^+$. 
\end{cond} 

In Condition~\ref{cond: vor an a(t)}, the requirement that $\scprd{ \,.\, , \,..\, }_t^+$ be a scalar product on $H^+$ whose norm $\norm{\,.\,}_t^+$ is equivalent to $\norm{\,.\,}^+$ for every $t \in I$ could be reformulated in an equivalent way by saying that there is $m \in (0,\infty)$ such that $a(t)(\,.\, , \,..\, ) + m \scprd{ \,.\, , \,..\, }$ is $\norm{\,.\,}^+$-bounded and $\norm{\,.\,}^+$-coercive. 
It is well-known that under Condition~\ref{cond: vor an a(t)} there is, for every $t \in I$, a unique self-adjoint operator $A_{a(t)}: D(A_{a(t)}) \subset H \to H$ such that 
\begin{align*}
D(A_{a(t)}) \subset H^+ \quad \text{and} \quad \scprd{x,A_{a(t)}y} = a(t)(x,y)
\end{align*}
for every $x \in H^+$ and $y \in D(A_{a(t)})$ (Theorem~VI.2.1 and Theorem~VI.2.6 of~\cite{KatoPerturbation80} or Theorem~10.1.2 of~\cite{BirmanSolomjak}).
As usual, we denote -- in the situation of Condition~\ref{cond: vor an a(t)} -- by $H^-$ the space of $\norm{\,.\,}^+$-continuous conjugate-linear functionals $H^+ \to \C$, which obviously is a complete space w.r.t.~the (mutually equivalent) norms 
\begin{align*}
f \mapsto \norm{f}^- := \sup_{\norm{x}^+ = 1} |f(x)| 
\quad \text{and} \quad 
f \mapsto \norm{f}_t^- := \sup_{\norm{x}_t^+ = 1} |f(x)| \quad (t \in I).
\end{align*}
It is straightforward to 
verify that these norms are induced by the scalar products $\scprd{\,.\,,\,..\,}^-$ and $\scprd{\,.\,,\,..\,}_t^-$ defined by
\begin{align} \label{eq:vorber, def sc prod}
\scprd{f,g}^- := \scprd{j^-(f),j^-(g)}^+
\quad \text{and} \quad
\scprd{f,g}_t^- := \scprd{j_t^-(f),j_t^-(g)}_t^+,
\end{align}
where $j^-: (H^-,\norm{\,.\,}^-) \to (H^+,\norm{\,.\,}^+)$ and $j_t^-: (H^-,\norm{\,.\,}_t^-) \to (H^+,\norm{\,.\,}_t^+)$ are the (unique) isometric isomorphisms with
\begin{align} \label{eq:vorber, def j^-}
f(x) = \scprd{x,j^-(f)}^+ \quad \text{and} \quad f(x) = \scprd{x,j_t^-(f)}_t^+ 
\quad (x \in H^+, f \in H^-).
\end{align}
We also denote by $j: H \to H^-$ the injective continuous linear map defined by $j(x) := \scprd{\,.\,,x} \in H^-$ for $x \in H$. It is straightforward to check using~\eqref{eq:vorber, def sc prod} and~\eqref{eq:vorber, def j^-} that the orthogonal complement of $j(H)$ 
in $H^-$ is trivial and hence $j(H)$ 
is dense in $H^-$. 
%
\smallskip

We continue by citing the 
fundamental theorem of Kisy\'{n}ski (Theorem~8.1 of~\cite{Kisynski63}) giving sufficient conditions for the well-posedness of the initial value problems corresponding to $A$ on $D(A(t))$, where $A(t) = 1/i A_{a(t)}$ with symmetric sesquilinear forms $a(t)$ with constant form 
domain. 
Similar theorems on well-posedness can be proved for the case 
of operators $A(t) = -A_{a(t)}$ defined by sectorial sesquilinear forms $a(t)$ with time-independent form 
domain. See, for instance, Fujie and Tanabe's article~\cite{FujieTanabe73} (Theorem~3.1) or Kato and Tanabe's article~\cite{KatoTanabe62} (Theorem~7.3). 

\begin{thm} [Kisy\'{n}ski] \label{thm: Kisynski}
Suppose $a(t)$ for every $t \in I$ is a sesquilinear form such that Condition~\ref{cond: vor an a(t)} is satisfied with $n = 2$ and set $A(t) := 1/i A_{a(t)}$ for $t \in I$. Then there is a unique evolution system $U$ for $A$ on $D(A(t))$ and $U(t,s)$ is unitary in $H$ for every $(s,t) \in \Delta$.
\end{thm}

In particular, this theorem guarantees that the basic Condition~\ref{cond: U_eps existiert und beschränkt} of the general adiabatic theorems for time-dependent domains is satisfied if only Condition~\ref{cond: vor an a(t)} is satisfied with $n = 2$. 
In verifying the other conditions of the general adiabatic theorems discussed in Section~\ref{sect:ad-thms-general}, the following lemma will be important. It allows us to express the resolvent of $A(t) := 1/i A_{a(t)}$ in terms of the inverse of an operator with time-independent domain and thus to verify the regularity conditions and estimates for the resolvent from our general adiabatic theorems. 
%

\begin{lm} \label{lm: allg lm}
Suppose that Condition~\ref{cond: vor an a(t)} is satisfied for a certain $n \in \N$ and, for every $t \in I$, denote by $\tilde{A}_0(t)$ the bounded linear operator $H^+ \to H^-$ extending $A_0(t):= A_{a(t)}$, that is, $\tilde{A}_0(t)x := a(t)(\,.\,,x)$ for $x \in H^+$.
Then the following holds true:
\begin{itemize}
\item[(i)] $t \mapsto \tilde{A}_0(t)$ is $n$ times weakly continuously differentiable.
\item[(ii)] If for a certain $z \in \C$ the operator $A_0(t)-z: D(A_0(t)) \subset H \to H$ is bijective for all $t \in J_0$ (a non-trivial subinterval of $I$), then so is $\tilde{A}_0(t)-z j: H^+ \to H^-$ and
\begin{align*}
(A_0(t)-z)^{-1} x = (\tilde{A}_0(t)- z j)^{-1} j(x)
\end{align*} 
for all $t \in J_0$ and $x \in H$.
In particular, $J_0 \ni t \mapsto (A_0(t)-z)^{-1}$ is $n$ times weakly continuously differentiable.
\end{itemize}
\end{lm}

\begin{proof}
(i) We have only to show that $t \mapsto F(\tilde{A}_0(t)x)$ is $n$ times continuously differentiable for every $x \in H^+$ and every $F \in (H^-)^*$, because $H^-$ is reflexive. (We point out that the reflexivity is essential here by the remarks after Definition~3.2.3 of~\cite{HillePhillips}.)
Since the canonical conjugate-linear map 
\begin{align*}
H^+ \ni y \mapsto i(y) \in (H^-)^* 
\text{ \, with \, } i(y)(f) := f(y) \text{ \, for } f \in H^- 
\end{align*}
is surjective by the reflexivity of $H^+$, the claim 
is obvious from the $n$ times continuous differentiability requirement in Condition~\ref{cond: vor an a(t)}. 
\smallskip   

(ii) We need some preparations. 
As a first preparatory step, we show that the operator $A_0^-(t): j(H^+) \subset H^- \to H^-$  defined by $A_0^-(t) j(x) := \tilde{A}_0(t)x$ for $x \in H^+$ is self-adjoint 
w.r.t.~$\scprd{\,.\,,\,..\,}_t^-$ and that 
\begin{align} \label{eq:vorber, step 1}
(A_0^-(t)-z)^{-1} j(x) = j\big( (A_0(t)-z)^{-1} x \big)
\end{align}
for all $z \in \C \setminus \R$ and $x \in H$.
We could refer to~\cite{Kisynski63} (Lemma~7.8) for the self-adjointness of $A_0^-(t)$, but for the reader's convenience we give an independent proof here. 
With the help of~\eqref{eq:vorber, def sc prod}, \eqref{eq:vorber, def j^-} and
\begin{align*}
\big( A_0^-(t)+m \big) j(x) = a(t)(\,.\,,x) + m \scprd{\,.\,,x} = \scprd{\,.\,,x}_t^+ 
\qquad (x \in H^+),
\end{align*}
it is straightforward to verify that 
\begin{align*}
\scprd{ \big( A_0^-(t)+m \big) j(x), j(y) }_t^- = j(y)(x) = \ol{j(x)(y)} = \scprd{ j(x), \big( A_0^-(t)+m \big) j(y) }_t^-
\end{align*}
for all $x, y \in H^+$. In other words, $A_0^-(t)+m$ and hence $A_0^-(t)$ is self-adjoint w.r.t.~the scalar product $\scprd{\,.\,,\,..\,}_t^-$. 
With the help of
\begin{align*}
\norm{ \big( A_0^-(t)+m \big) j(x) }_t^- = \sup_{\norm{y}_t^+ = 1} |\scprd{y,x}_t^+| = \norm{x}_t^+ \qquad (x \in H^+),
\end{align*}
it is also easy to see that $A_0^-(t)$ is a closed operator in $H^-$. 
Since
\begin{align} \label{eq:vorber, step 1, 1}
A_0^-(t)j(y) = a(t)(\,.\,,y) = \scprd{\,.\,,A_0(t)y} = j(A_0(t)y) 
\qquad (y \in D(A_0(t))
\end{align}
and since $\ran (A_0(t) \pm i) = H$ by the self-adjointness of $A_0(t)$ in $H$, it finally follows that 
\begin{align*}
\ran(A_0^-(t) \pm i) \supset j\big( \ran (A_0(t) \pm i) \big) = j(H)
\end{align*}
is dense in $H^-$. 
So, by the basic criterion for self-adjointness (Theorem~VIII.3 of~\cite{ReedSimon}), we see that $A_0^-(t)$ is self-adjoint in $(H^-,\norm{\,.\,}_t^-)$. In particular, $\C \setminus \R \subset \rho(A_0^-(t))$ and thus~\eqref{eq:vorber, step 1} follows by~\eqref{eq:vorber, step 1, 1}. 
As a second preparatory step, we show that
\begin{align} \label{eq:vorber, step 2}
\rho(A_0(t)) \subset \rho(A_0^-(t)) 
\end{align}
for every $t \in I$. 
Since $\C \setminus \R \subset \rho(A_0^-(t))$ by the first step, it suffices to prove that $\rho(A_0(t)) \cap \R \subset \rho(A_0^-(t))$. 
So, let $\lambda \in \rho(A_0(t)) \cap \R$, then there is a $\delta > 0$ such that 
$(z- 2 \delta, z+ 2 \delta) \subset \rho(A_0(t))$. It thus follows by Stone's formula 
(applied to both $A_0(t)$ and $A_0^-(t)$) and~\eqref{eq:vorber, step 1}  that
\begin{align*}
0 = j\Big( P_{(\lambda-\delta,\lambda+\delta)}x + \frac{1}{2} P_{ \{ \lambda-\delta, \lambda+\delta\} }x \Big) = \Big( P_{(\lambda-\delta,\lambda+\delta)}^- + \frac{1}{2} P_{ \{ \lambda-\delta, \lambda+\delta\} }^- \Big) j(x)
\end{align*}
for all $x \in H$, where $P$ and $P^-$ denote the spectral measure of $A_0(t)$ and $A_0^-(t)$, respectively. So, by the density of $j(H)$ in $H^-$, 
\begin{align*}
0 = P_{(\lambda-\delta,\lambda+\delta)}^- + \frac{1}{2} P_{ \{ \lambda-\delta, \lambda+\delta\} }^- 
\ge P_{(\lambda-\delta,\lambda+\delta)}^-
\ge 0
\end{align*}
and therefore 
$\lambda \in \C \setminus \operatorname{supp}(P^-) = \rho(A_0^-(t))$, as desired.
With the above preparations at hand, we can now easily conclude the proof. Indeed, let $z \in \C$ and $J_0$ be a non-trivial subinterval of $I$ such that $A_0(t)-z: D(A_0(t)) \subset H \to H$ is bijective for every $t \in J_0$. It then follows by the second step that $A_0^-(t)-z: j(H^+) \subset H^- \to H^-$ and hence $\tilde{A}_0(t)-z j: H^+ \to H^-$ is bijective as well. Also, by~\eqref{eq:vorber, step 1, 1} the claimed formula for $(A_0(t)-z)^{-1}$ follows and from this formula, in turn, we get the claimed $n$ times weak continuous differentiability of $J_0 \ni t \mapsto  (A_0(t)-z)^{-1}$ using part~(i) and the inverse rule for weak continuous differentiability (Section~1.4.2 of~\cite{Kisynski63}, for instance). 
\end{proof}

\subsection{Adiabatic theorems with spectral gap condition}

We will need the following condition depending on a parameter $m \in \{0\} \cup \N \cup \{ \infty \}$ for the adiabatic theorem with spectral gap condition below.

\begin{cond} \label{cond: vor adsatz mit sl für A(t)=iA_{a(t)}}
$A(t) = 1/i A_{a(t)}$ for $t \in I$, where the sesquilinear forms $a(t)$ satisfy Condition~\ref{cond: vor an a(t)} with $n = 2$. 
Also, 
$\sigma(t)$ for every $t \in I$ is a compact subset of $\sigma(A(t))$, 
$\sigma(\,.\,)$ falls into $\sigma(A(\,.\,)) \setminus \sigma(\,.\,)$ at exactly $m$ points that accumulate at only finitely many points, and $I \setminus N \ni t \mapsto \sigma(t)$ is continuous, where $N$ denotes the set of those $m$ points at which $\sigma(\,.\,)$ falls into $\sigma(A(\,.\,)) \setminus \sigma(\,.\,)$.
And finally, 
$P(t)$ for every $t \in I \setminus N$ is the projection associated with $A(t)$ and $\sigma(t)$ and $I \setminus N \ni t \mapsto P(t)$ extends to a twice strongly continuously differentiable map (again denoted by $P$) on the whole of $I$.
\end{cond}


In view of Lemma~\ref{lm: allg lm} it is 
easy to derive the following adiabatic theorem with uniform ($m =0$) or non-uniform ($m \in \N \cup \{ \infty \}$) spectral gap condition from the corresponding general adiabatic theorem with spectral gap condition 
(Theorem~\ref{thm: adsatz mit sl, zeitabh}).

\begin{thm} \label{thm: adsatz mit sl}
Suppose $A(t)$, $\sigma(t)$, $P(t)$ for $t \in I$ are as in Condition~\ref{cond: vor adsatz mit sl für A(t)=iA_{a(t)}} with $m = 0$ or with $m \in \N \cup \{\infty\}$, respectively. Then
\begin{align*}
\sup_{t \in I} \norm{ U_{\eps}(t) - V_{\eps}(t) } = O(\eps)  \text{ resp. } o(1) \quad (\eps \searrow 0), 
\end{align*}
whenever the evolution system $V_{\eps}$ for $\frac 1 \eps A + [P',P]$ exists on $D(A(t))$ for every $\eps \in (0, \infty)$.
\end{thm}

\begin{proof}
Choose, for every $t_0 \in I \setminus N$, non-trivial closed intervals $J_{t_0}$ and cycles $\gamma_{t_0}$ as in Condition~\ref{cond: vor adsatz mit sl}. 
Since $I \setminus N$ is relatively open in $I$, such choices are possible. 
Lemma~\ref{lm: allg lm} now shows that Condition~\ref{cond: vor adsatz mit sl} is satisfied and, therefore, the assertion follows from Theorem~\ref{thm: adsatz mit sl, zeitabh}.
\end{proof}

If the existence of the evolution $V_{\eps}$ for $\frac 1 \eps A +[P',P]$ cannot be ensured, 
one can still exploit 
the remark after Theorem~\ref{thm: adsatz mit sl, zeitabh}.
In the case of a uniform spectral gap, 
the existence of $V_{\eps}$ is guaranteed if, for instance, Condition~\ref{cond: vor an a(t)} is 
fulfilled with $n = 3$. Indeed, in that case $I \ni t \mapsto P(t)$ is thrice weakly continuously differentiable (by Lemma~\ref{lm: allg lm}~(ii)) 
so that the symmetric sesquilinear forms $\frac 1 \eps a(t) + b(t) := \frac 1 \eps a(t) +i \scprd{\,.\,,[P'(t),P(t)]\,..\,}$ corresponding to $\frac{1}{\eps}A(t) + [P'(t),P(t)]$ satisfy Condition~\ref{cond: vor an a(t)} with $n = 2$ and therefore Theorem~\ref{thm: Kisynski} can be applied. 
\smallskip

We finally note conditions under which the general adiabatic theorem of higher order (Theorem~\ref{thm: höherer adsatz}) can be applied to the case  of operators $A(t)$ defined by symmetric sesquilinear forms.

\begin{cond} \label{cond: adsatz höherer ord}
Suppose that $A(t) = 1/i A_{a(t)}$ for $t \in I$ where the sesquilinear forms $a(t)$ satisfy Condition~\ref{cond: vor an a(t)} with a certain $n \in \N \setminus \{1\}$ 
or with $n = \infty$, respectively. 
In the latter case suppose further that there is an open neighbourhood $U_I$ of $I$ in $\C$ and for each $w \in U_I$ there is a $\norm{\,.\,}^+$-bounded sesquilinear form $\tilde{a}(w)$ on $H^+$ such that $\tilde{a}(t) = a(t)$ for $t \in I$ and that $$U_I \ni w \mapsto \tilde{a}(w)(x,y)$$ is holomorphic for every $x,y \in H^+$. 
Suppose moreover that $\sigma(t)$ for every $t \in I$ is an isolated compact subset of $\sigma(A(t))$, that $\sigma(\,.\,)$ at no point falls into $\sigma(A(\,.\,)) \setminus \sigma(\,.\,)$, 
and that $t \mapsto \sigma(t)$ is continuous. 
And finally, suppose $P(t)$ for every $t \in I$ is the projection associated with $A(t)$ and $\sigma(t)$ and $t \mapsto P(t)$ is $n+1$ times 
times strongly continuously differentiable.
\end{cond}



It is not difficult -- albeit a bit technical -- to show that  under Condition~\ref{cond: adsatz höherer ord} the hypotheses of Theorem~\ref{thm: höherer adsatz} are really satisfied.
(In the case $n = \infty$ define $\tilde{A}_0(w)$ by $\tilde{A}_0(w)x := \tilde{a}(w)(\,.\, ,x)$ for $x \in H^+$. Then $\tilde{A}_0(w)$ is a bounded linear operator $H^+ \to H^-$ and  
$U_I \ni w \mapsto \tilde{A}_0(w) \in L(H^+,H^-)$ is weakly holomorphic and hence holomorphic w.r.t.~the norm operator topology. A simple perturbation argument and Cauchy's inequality (in conjunction with the formula in Lemma~\ref{lm: allg lm}~(ii)) then yield the estimates required in Condition~\ref{cond: vor höherer adsatz}.)

%

\subsection{An adiabatic theorem without spectral gap condition}  \label{sect: adsatz ohne sl für slinformen}

In the adiabatic theorem without spectral gap condition below, the following condition will be used. 

\begin{cond} \label{cond: vor adsatz ohne sl}
$A(t) = 1/i A_{a(t)}$ for $t \in I$ where the sesquilinear forms $a(t)$ satisfy Condition~\ref{cond: vor an a(t)} with $n = 2$. 
Also, 
$\lambda(t)$ for every $t \in I$ is an eigenvalue of $A(t)$ such that $t \mapsto \lambda(t)$ is continuous. 
And finally, $P(t)$ for every $t \in I$ is an orthogonal projection in $H$ such that $P(t)$ is weakly associated with $A(t)$ and $\lambda(t)$ for almost every $t \in I$, $\rk P(0) < \infty$ and $t \mapsto P(t)$ is strongly continuously differentiable.
\end{cond}

While in the case with spectral gap Lemma~\ref{lm: allg lm} was sufficient, 
we need an additional lemma in the case without spectral gap. 

\begin{lm} \label{lm: lm für adsatz ohne sl}
Suppose that Condition~\ref{cond: vor adsatz ohne sl} is satisfied and that, in addition, $t \mapsto \lambda(t)$ is even continuously differentiable. 
Then $t \mapsto (\lambda(t) + \delta - A(t))^{-1}$ is strongly continuously differentiable for every $\delta \in (0,\infty)$ and there is an $M_0' \in (0,\infty)$ such that
\begin{align*}
\norm{ \ddt{ (\lambda(t) + \delta - A(t))^{-1} }  } \le \frac{M_0'}{\delta^2}
\end{align*}
for all $t \in I$ and $\delta \in (0,1]$.
\end{lm}

\begin{proof}
Set $A_0(t) := A_{a(t)} = i A(t)$ and $\lambda_0(t) := i \lambda(t)$ and let $\tilde{A}_0(t): H^+ \to H^-$ be the bounded extension of $A_0(t)$. 
Since by Lemma~\ref{lm: allg lm} $t \mapsto \tilde{A}_0(t)$ is twice weakly and, in particular, once strongly continuously differentiable and $t \mapsto \lambda_0(t)$ is continuously differentiable, it follows that 
\begin{align*}
t \mapsto \big( A_0(t) - (\lambda_0(t) + i\delta) \big)^{-1} = \big( \tilde{A}_0(t) - (\lambda_0(t) + i \delta) j \big)^{-1} j
\end{align*}
is strongly continuously differentiable for every $\delta \in (0,\infty)$ and that
\begin{align} \label{eq: darst der abl}
&\ddt{ \big( A_0(t) - (\lambda_0(t) + i\delta) \big)^{-1} } \notag \\
&\qquad \quad = \big( \tilde{A}_0(t) - (\lambda_0(t) + i \delta) j \big)^{-1} \, \big( \lambda_0'(t)j -\tilde{A}_0'(t) \big) \, \big( \tilde{A}_0(t) - (\lambda_0(t) + i \delta) j \big)^{-1} j
\end{align}
for $t \in I$ and $\delta \in (0,\infty)$. 
We now show 
that there is a constant $c_0 \in (0,\infty)$ such that 
\begin{align} \label{eq: wesentl zwbeh}
\norm{ \big( \tilde{A}_0(t) - (\lambda_0(t) + i \delta)j \big) x }_t^- \ge \frac{ \delta}{c_0} \norm{x}_t^+
\end{align}
for all $x \in H^+$, $t \in I$ and $\delta \in (0,1]$. In order to do so, we observe the following simple fact: if instead of $j$ the natural isometric isomorphism 
\begin{align*}
j_t^+: (H^+, \norm{\,.\,}_t^+) \to (H^-, \norm{\,.\,}_t^-) \text{\, with \,} j_t^+(x) := \scprd{\,.\,,x}_t^+ \text{ for } x \in H^+ 
\end{align*}
occurred in~\eqref{eq: wesentl zwbeh}, this assertion would be trivial. 
We are therefore led to express $j$ in terms of $j_t^+$: by the definition of the scalar product $\scprd{ \,.\,,\,..\, }_t^+$ in Condition~\ref{cond: vor an a(t)}, we have 
\begin{align*}
j|_{H^+} = \frac{1}{m} \big( j_t^+ - \tilde{A}_0(t) \big)
\end{align*} 
for all $t \in I$, so that 
\begin{align} \label{eq:vorber-ohne-sl, zwrechnung}
\tilde{A}_0(t) - (\lambda_0(t) + i \delta)j = \frac{m+\lambda_0(t)+i \delta}{m} \Big( \tilde{A}_0(t) - \frac{ \lambda_0(t) +i\delta}{m + \lambda_0(t) +i\delta} j_t^+ \Big).
\end{align}
Since for all $x \in H^+$ with $\norm{x}_t^+ = 1$
\begin{align*}
\norm{  \Big( \tilde{A}_0(t) - \frac{ \lambda_0(t) +i\delta}{m + \lambda_0(t) +i\delta} j_t^+ \Big)x  }_t^- &\ge \Big| a(t)(x,x) - \frac{ \lambda_0(t) +i\delta}{m + \lambda_0(t) +i\delta} \big( j_t^+(x) \big)(x) \Big| \\
&\ge \Big| \Im \Big( \frac{ \lambda_0(t) +i\delta}{m + \lambda_0(t) +i\delta} \Big) \Big| = \frac{m \delta}{ |m + \lambda_0(t) +i\delta|^2 },
\end{align*}
it follows by~\eqref{eq:vorber-ohne-sl, zwrechnung} that
\begin{align*}
\norm{ \big( \tilde{A}_0(t) - (\lambda_0(t) + i \delta)j \big)x }_t^- \ge \Big| \frac{m+\lambda_0(t)+i \delta}{m} \Big| \, \frac{m \delta}{ |m + \lambda_0(t) +i\delta|^2 } \norm{x}_t^+ \ge \frac{\delta}{c_0} \norm{x}_t^+  
\end{align*}
for all $x \in H^+$ and all $t \in I$, $\delta \in (0,1]$, where $c_0 := m + \norm{\lambda}_{\infty} + 1$. So \eqref{eq: wesentl zwbeh} is proven and it follows that
\begin{align} \label{eq: absch inverse}
\norm{ \big( \tilde{A}_0(t) - (\lambda_0(t) + i \delta)j \big)^{-1} }_{H^-, H^+} \le \frac{c_0'}{ \delta}
\end{align}
for all $t \in I$ and $\delta \in (0,1]$, because the equivalence of the norms $\norm{\,.\,}_t^+$ with $\norm{\,.\,}$ required in Condition~\ref{cond: vor an a(t)} is uniform w.r.t.~$t$ by Lemma~7.3 of~\cite{Kisynski63}.  
In view of~\eqref{eq: darst der abl} and~\eqref{eq: absch inverse} 
the asserted estimate is now clear. 
\end{proof}

With this lemma at hand, it is now simple to derive the following adiabatic theorem without spectral gap condition, which 
generalizes an adiabatic theorem of Bornemann (Theorem~IV.1 of~\cite{Bornemann98}). See the discussion below for a detailed 
comparison of these results.

\begin{thm} \label{thm: adsatz ohne sl für A(t)=iA_{a(t)}}
Suppose $A(t)$, $\lambda(t)$, $P(t)$ for $t \in I$ are such that Condition~\ref{cond: vor adsatz ohne sl} is satisfied.
Then
\begin{align*}
\sup_{t \in I} \norm{ \big( U_{\eps}(t) - V_{0\,\eps}(t) \big) P(0) } \longrightarrow 0 \quad \text{and} \quad \sup_{t \in I} \norm{ P(t) \big( U_{\eps}(t) - V_{0\,\eps}(t) \big) } \longrightarrow 0 
\end{align*}
as $\eps \searrow 0$, where $V_{0\,\eps}$ denotes the evolution system for $\frac{1}{\eps}AP + [P',P] = \frac 1 \eps \lambda P + [P',P]$ for every $\eps \in (0,\infty)$. 
If, in addition, $t \mapsto P(t)$ is thrice weakly continuously differentiable, then the evolution system $V_{\eps}$ for $\frac 1 \eps A + [P',P]$ exists on $D(A(t))$ for every $\eps \in (0,\infty)$ and
\begin{align*}
\sup_{t \in I} \norm{U_{\eps}(t) - V_{\eps}(t) } \longrightarrow 0 \quad (\eps \searrow 0).
\end{align*}
\end{thm}

\begin{proof}
We have to verify the hypotheses of the general adiabatic theorem without spectral gap condition for time-dependent domains (Theorem~\ref{thm: erw adsatz ohne sl, zeitabh}) with $m_0 = 1$. In view of Lemma~\ref{lm: lm für adsatz ohne sl} it remains to establish three small things, namely the inclusions $P(t)H \subset \ker(A(t)-\lambda(t))$ and $P(t)A(t) \subset A(t)P(t)$ for every $t \in I$ (required in Theorem~\ref{thm: erw adsatz ohne sl, zeitabh}) and the continuous differentiability of $t \mapsto \lambda(t)$ (required in  Theorem~\ref{thm: erw adsatz ohne sl, zeitabh} and Lemma~\ref{lm: lm für adsatz ohne sl}). 
We know by assumption that $P(t)H = \ker(A(t)-\lambda(t)) = \ker(A_0(t)-\lambda_0(t))$ for almost every $t \in I$ so that $P(t)H \subset D(A_0(t)) \subset H^+$ and
\begin{align*}
0 = j\big( (A_0(t)-\lambda_0(t))P(t)x \big) 
= \big( A_0^-(t)-\lambda_0(t) \big) j(P(t)x)
\end{align*} 
for all $x \in H$ and almost every $t \in I$ (where $A_0(t)$, $\lambda_0(t)$ are defined as in the proof of Lemma~\ref{lm: lm für adsatz ohne sl} and where $A_0^-(t)$ is the self-adjoint operator in $(H^-, \norm{\,.\,}_t^-)$ from the proof of Lemma~\ref{lm: allg lm}). Applying the  closedness argument 
after Theorem~3.2 of~\cite{zeitunabh} 
to the closed 
operator $1/i A_0^-(t): j(H^+) \subset H^- \to H^-$ (with time-independent domain!), 
we see that $j(P(t)H) \subset j(H^+)$ 
and 
\begin{align*}
0 = \big( A_0^-(t)-\lambda_0(t) \big) j(P(t)x) = a(t)(\,.\,,P(t)x) - \lambda_0(t) \scprd{\,.\,,P(t)x}
\end{align*}  
for all $x \in H$ and every (not only almost every) $t \in I$. In particular, for every $t \in I$, 
\begin{align*}
0= a(t)(y,P(t)x) - \lambda_0(t) \scprd{y,P(t)x} = \scprd{(A_0(t)-\lambda_0(t))y, P(t)x}
\end{align*}
for $y \in D(A_0(t))$ and $x \in H$, so that 
\begin{align} \label{eq: inklusion für alle t}
P(t)H \subset \ker(A_0(t)-\lambda_0(t))^* = \ker(A(t)-\lambda(t))
\end{align}
for every $t \in I$, as desired.
In other words, $A(t)P(t) = \lambda(t)P(t)$ for every $t \in I$ and therefore we also obtain
\begin{align*}
P(t)A(t) = -P(t)^* A(t)^* \subset -\big( A(t)P(t) \big)^* 
= \lambda(t)P(t) = A(t)P(t)
\end{align*}
for every $t \in I$, as desired.
Since, finally, for every $t_0 \in I$ there is a neighbourhood $J_{t_0} \subset I$ and an $x_0 \in H$ such that $P(t)x_0 \ne 0$ for $t \in J_{t_0}$, 
it follows from~\eqref{eq: inklusion für alle t} that
\begin{align*}
\frac{1}{\lambda(t)-1} = \frac{   \scprd{P(t)x_0, (A(t)-1)^{-1} P(t)x_0}  }{   \scprd{P(t)x_0, P(t)x_0}   } 
\end{align*} 
for every $t \in J_{t_0}$, from which in turn it follows (by Lemma~\ref{lm: allg lm}) 
that $t \mapsto \lambda(t)$ is continuously differentiable, as desired.
\smallskip

According to what has been said at the beginning of the proof, it is now clear that Lemma~\ref{lm: lm für adsatz ohne sl} can be applied and that the hypotheses of the first part of Theorem~\ref{thm: erw adsatz ohne sl, zeitabh} are satisfied. 
Since the evolution system $U_{\eps}$ is unitary (by Theorem~\ref{thm: Kisynski}) and $V_{0\,\eps}$ is unitary as well (by Theorem~X.69 of~\cite{ReedSimon}), we see by modifying the proof of Theorem~\ref{thm: erw adsatz ohne sl, zeitabh} in an obvious manner that
\begin{align} \label{eq: aussage adsatz für unitäre zeitentw}
\sup_{(s,t)\in I^2} \norm{ \big( U_{\eps}(t,s) - V_{0\,\eps}(t,s) \big) P(s) } \longrightarrow 0 \quad (\eps \searrow 0),
\end{align} 
where $U_{\eps}(t,s) := U_{\eps}(s,t)^{-1} = U_{\eps}(s,t)^*$ and $V_{0\,\eps}(t,s) := V_{0\,\eps}(s,t)^{-1} = V_{0\,\eps}(s,t)^*$ for $(s,t) \in I^2$ with $s > t$. Since 
\begin{align*}
\norm{ P(t) \big( U_{\eps}(t) -V_{0\,\eps}(t) \big) } = \norm{ \big(U_{\eps}(0,t)-V_{0\,\eps}(0,t)\big)P(t) }
\end{align*}
for $t \in I$ (take adjoints), the first two of the asserted convergences follow from~\eqref{eq: aussage adsatz für unitäre zeitentw}. 
\smallskip

Suppose finally that $t \mapsto P(t)$ is thrice weakly continuously differentiable. Then the symmetric sesquilinear forms $$\frac 1 \eps a(t) + b(t) := \frac 1 \eps a(t) +i \scprd{\,.\,,[P'(t),P(t)]\,..\,}$$ corresponding to the operators $\frac{1}{\eps}A(t) + [P'(t),P(t)]$ satisfy Condition~\ref{cond: vor an a(t)} with $n = 2$ and therefore the evolution system $V_{\eps}$ for $\frac 1 \eps A + [P',P]$ exists on $D(A(t))$ for every $\eps \in (0,\infty)$ by virtue of Theorem~\ref{thm: Kisynski}. Also, $t \mapsto P(t)$ is obviously norm continuously differentiable and so the hypotheses of the second part of Theorem~\ref{thm: erw adsatz ohne sl, zeitabh} are satisfied, which gives the third and last 
of the asserted convergences.
\end{proof}

What are the differences between the above theorem and Bornemann's adiabatic theorem of~\cite{Bornemann98}? 
While in Theorem~IV.1 of~\cite{Bornemann98} $\lambda(t)$ is required to belong to the discrete spectrum of $A(t)$ (and hence to be an isolated eigenvalue) for every $t \in I$, 
in the above theorem it is only required that $\lambda(t)$ has finite multiplicity for almost every $t \in I$: the eigenvalues $\lambda(t)$ are allowed to have infinite multiplicity on a set of measure zero and, moreover, they are allowed to be non-isolated in $\sigma(A(t))$ for every $t \in I$.
Also, the regularity conditions on $A$ and $P$ of the above theorem are slightly weaker than those of Theorem~IV.1: for instance, $t \mapsto \tilde{A}_0(t)$ is required to be twice continuously differentiable w.r.t.~the norm operator topology in~\cite{Bornemann98} whereas above it is only required that $t \mapsto a(t)(x,y)$ be twice continuously differentiable for $x, y \in H^+$ (or equivalently (Lemma~\ref{lm: allg lm}), that $t \mapsto \tilde{A}_0(t)$ be twice weakly continuously differentiable). 
%
And finally, the assertion of the theorem above is 
more general than the conclusion of Theorem~IV.1 in~\cite{Bornemann98} which says that, for all $x \in H^+$ (and hence for all $x \in H$)  
and uniformly in $t \in I$,
\begin{align*}
& \scprd{ U_{\eps}(t)x, P(t)U_{\eps}(t)x } 
 = \scprd{ U_{\eps}(t)x, P(t)U_{\eps}(t)x - U_{\eps}(t)P(0)x } + \scprd{ x, P(0)x } \\
& \qquad \qquad \qquad \qquad \quad \longrightarrow  \scprd{x, P(0)x} \quad (\eps \searrow 0). 
\end{align*} 

\section*{Acknowledgement}

I would like to thank Marcel Griesemer for numerous discussions and for introducing me to adiabatic theory in the first place. I would also like to thank the German Research Foundation (DFG) for financial support through the research training group ``Spectral theory and dynamics of quantum systems'' (GRK 1838).

\begin{small}
%
%
%


\end{small}

\end{document}